\pdfoutput=1
\documentclass[draftclsnofoot, onecolumn, 12pt]{IEEEtran}
\usepackage{mathtools} % \allowdisplaybreaks, \coloneqq, ...
\usepackage{amsthm, amssymb} % theorems and symbols
\usepackage{bm} % bold vectors
\usepackage{pgfplots} % plots
\usepackage[hidelinks]{hyperref} % cross-referencing and hypertext links, bookmarks, ...
\usepackage{microtype} % micro-typographic adjustments
\usepackage[doi = false, isbn = false, url = false, style = ieee]{biblatex} % bibliography management
\usepackage{orcidlink} % ORCID symbol (external style file on Overleaf)
\usepackage{siunitx} % percentages
\usepackage{booktabs} % for tables
\usepackage[inline]{enumitem} % for inline itemize
\usepackage[nameinlink]{cleveref} % references

\pgfplotsset{compat = newest}
\usetikzlibrary{shapes, calc, decorations.pathreplacing}
\allowdisplaybreaks

\theoremstyle{plain}
\newtheorem{theorem}{Theorem}
\newtheorem{proposition}{Proposition}
\newtheorem{lemma}{Lemma}
\newtheorem{corollary}{Corollary}[theorem]

\theoremstyle{definition}
\newtheorem{definition}{Definition}
\newtheorem{example}{Example}

\theoremstyle{remark}
\newtheorem{remark}{Remark}

\newcommand{\refappendix}[1]{\hyperref[#1]{Appendix~\ref*{#1}}}

\newcommand{\Wtwo}[3]{W_{f_{\hat{u}_#1}, \mathcal{D}_{\hat{u}_#1}, \{\hat{u}_#2, \hat{u}_#3\}}}
\newcommand{\Wthree}[4]{W_{f_{\hat{u}_#1}, \mathcal{D}_{\hat{u}_#1}, \{\hat{u}_#2, \hat{u}_#3, \hat{u}_#4\}}}

\bibliography{references.bib}

\title{Unselfish Coded Caching can Yield Unbounded Gains over Symmetrically Selfish Caching}

\author{Federico Brunero\textsuperscript{\orcidlink{0000-0002-6980-3827}} and Petros Elia\textsuperscript{\orcidlink{0000-0002-3531-120X}}%
  \thanks{%
    This work was supported by the European Research Council (ERC) through the EU Horizon 2020 Research and Innovation Program under Grant 725929 (Project DUALITY).

    The authors are with the Communication Systems Department at EURECOM, 450 Route des Chappes, 06410 Sophia Antipolis, France (email: brunero@eurecom.fr; elia@eurecom.fr).
  }
}

\begin{document}

\maketitle

\begin{abstract}
The original coded caching scenario assumes a content library that is of interest to all receiving users. In a realistic scenario though, the users may have diverging interests which may intersect to various degrees. What happens for example if each file is of potential interest to, say, $\SI{40}{\percent}$ of the users and each user has potential interest in $\SI{40}{\percent}$ of the library? In this work, we investigate the so-called \emph{symmetrically selfish coded caching} scenario, where each user only makes requests from a subset of the library that defines its own \emph{File Demand Set (FDS)}, each user caches selfishly only contents from its own FDS, and where the different FDSs symmetrically overlap to some extent. In the context of various traditional prefetching scenarios (prior to the emergence of coded caching), selfish approaches were known to be potentially very effective. On the other hand --- with the exception of some notable works --- little is known about selfish coded caching. We here present a new information-theoretic converse that proves, in a general setting of symmetric FDS structures, that selfish coded caching, despite enjoying a much larger local caching gain and a much smaller set of possible demands, introduces an unbounded load increase compared to the unselfish case. In particular, in the $K$-user broadcast channel where each user stores a fraction $\gamma$ of the library, where each file (class) is of interest to $\alpha$ users, and where any one specific file is of interest to a fraction $\delta$ of users, the optimal coding gain of symmetrically selfish caching is at least $(K - \alpha)\gamma + 1$ times smaller than in the unselfish scenario. This allows us to draw the powerful conclusion that the optimal selfish coding gain is upper bounded by $1/(1 - \delta)$, and thus does not scale with $K$. These derived limits are shown to be exact for different types of demands.
\end{abstract}


\section{Introduction}
\nocite{6763007}

\IEEEPARstart{T}{he} vast increase of network traffic has sparked considerable interest in finding new techniques that reduce the communication load. Toward this, caching has been traditionally used to bring contents closer to their destinations, thus reducing the volume of the communication problem during peak hours~\cite{5461964}. A key ingredient in using caches has commonly been the exploitation of the fact that some contents/files are more popular than others, and thus are generally to be allocated more cache space~\cite{ZHANG20133128, 6600983}. This inevitably introduces the consideration that different users may have different file preferences, which in turn brings to the fore the concept of \emph{selfish caching} where simply users cache independently and selfishly only contents that they are interested in potentially consuming themselves~\cite{8849537, 9148779, 8849357, 9149113}. In the traditional prefetching scenario where emphasis is based heavily on bringing relevant content closer to each user, this idea of selfish caching brought about performance improvements \cite{7485863, 9098073} in the form of higher local caching gains for each user.

A completely different utilization of caching was witnessed with the advent of coded caching~\cite{6763007}, whose focus is more on leveraging storage capabilities in order to reduce interference. Depending on the network topology, this coded variant can be a more powerful approach than traditional prefetching, because it employs caching not only to change the \emph{volume} of the communication problem, but also to change the \emph{structure} of the problem itself, simply by changing the interference patterns. Coded caching has been rightfully credited with being able to transform memory into data rates, and has hence sparked a flurry of research on a variety of topics such as on the interplay between caching and PHY~\cite{7580630, 8950279, shariatpanahi2018multiserver, 7864374, 8007072,8036265,9149433,8374074,naderializadehInterferenceManagementTransIT2017,ShariatpanahiPhysicalLayer2019TransIT,9007518}, caching and privacy~\cite{7959863, 9330765, 9249022}, on information-theoretic converses~\cite{8963629, 8226776}, on the critical bottleneck of subpacketization~\cite{yanPDATransIT2017,tangSubpacketizationTransIT2018,krishnanCCBipartiteGraphs2018, shangguanHypergraphsTransIT2018,8374958}, and a variety of other scenarios~\cite{8977539,8911362,9163148,9036921,8278044,9328826}.

In trying to fuse the traditional caching techniques with coded caching, a variety of works has naturally sought to explore coded caching in the presence of files with different popularity. This is an area of active research that has produced several interesting and insightful results~\cite{7782760, 8091300, 7904696, quinton2018novel, 8422960, 8322577, 8863425, 7994812, deng2021fundamental, 7865913, deng2021memoryrate} that focus on the scenario where the file popularity profiles are identical for every user.

\subsection{Heterogeneous User Profiles and Selfish Coded Caching}

On the other hand, we are just beginning to explore the connection between coded caching and selfish caching, where by selfish caching we generally refer to caching schemes in which each user caches only contents that meet its own individual preferences and objectives.

Recent works have sought to explore this connection. For example, in the context of coded caching with users having heterogeneous content preferences, the recent work in~\cite{9264196} took a game theoretic perspective to analyze the performance of coded caching when it accounts for this heterogeneity. Employing interesting analysis, this work revealed gains from taking this heterogeneity into consideration, where these gains were naturally a function of the structure of the user preferences. Furthermore, the work in~\cite{8849537} analyzed the peak load of three different coded caching schemes that account for the user preferences, and again revealed occasional performance gains that are similarly dependent on the structure of these preferences. Related analysis appears in~\cite{9148779}, now for the average load of these same schemes in~\cite{8849537}.

On the other hand, the work in~\cite{8849357} focused on finding instances where unselfish coded caching outperforms selfish designs. This work nicely considered the performance of selfish coded caching in the context of heterogeneous file demand sets, cleverly employing bounds to show that, for the case of $K = 2$ users and $3$ files, unselfish designs strictly outperform selfish designs in terms of communication load, albeit only by a factor of $\SI{14}{\percent}$. In addition, the notable work in~\cite{9149113} established the optimal average load --- under the assumption of selfish and uncoded prefetching --- for the case of $K = 2$ users and a variety of overlaps between the two users' profiles, also providing explicit prefetching schemes of a selfish nature. To the best of our understanding, the above constitutes the extent of works on selfish coded caching.

\subsection{An Adversarial Interplay between Coded Caching and Selfish Caching}

Our motivation to understand the interplay between coded caching and selfish caching comes not only from the fact that coded caching systems may indeed need to operate under some selfish legacy constraints\footnote{Here we can think of a scenario where a server delivers --- via a bottleneck link --- content to caches, whose purpose is to bring content closer to the end user via dedicated non-interfering links.  In such scenario, the delivery \emph{to} the caches would benefit from a coded caching design, while the subsequent delivery \emph{from} the caches would benefit from a selfish placement since the caches may target groups of users with potentially dissimilar interests.}, but also mainly from the fact that there exists an interesting ``adversarial'' interplay between coded caching and selfish prefetching. To understand this a bit better, we recall that the main idea of coded caching is that it multicasts at any given time a linear combination of different contents desired by different users. This implies that any one receiver associated to a multicast message must be able to find in its cache all the undesired contents (subfiles) of that multicast message. This is achieved in \cite{6763007} by means of a highly structured and coordinated content placement phase, where each user caches a small fraction of \emph{every} file of a common library. This relationship between undesired and cached contents deteriorates when using selfish caching, simply because each receiver selfishly opts --- based on its own preferences --- to not cache some of these undesired files. These same undesired files though may eventually appear as interference at that selfish receiver who will now not be able to ``cache-out'' this interference. At the same time though, such selfish caching allows for a much more targeted placement of files such that each user can cache more of what it actually wants. Furthermore, such selfish scenario would correspond to a substantially smaller set of possible demands, which could conceivably be exploited to reduce the load.

\subsection{Main Contributions}

To understand this interplay between coded caching and selfish caching, we first propose a new selfish model which aims to calibrate the selfishness effect, by calibrating the degree of separation between the interests of the different users.  Our so-called symmetric File Demand Set (FDS) structure not only aims to encapsulate this aspect of intersection of interests, but is also designed to reflect and accentuate the aforementioned adversarial relationship between the selfish placement and the ability to encode across users as one would expect in the coded caching setting. 

Then, for the aforementioned symmetric FDS structure, we employ index coding arguments to derive an information-theoretic converse (lower bound) on the optimal worst-case communication load under the assumption of uncoded and selfish placement. This bound proves that generally unselfish coded caching far outperforms selfish coded caching. The bound makes clear the fact that, while, as noted, selfish caching implies a much smaller set of possible demands (cf.~\Cref{def: Uncoded and Selfish Cache Placement Definition}) as well as allows for a much more targeted placement of contents, these benefits come at a heavy cost of fewer multicasting opportunities and a substantial loss in coding gain.
The main contribution of our work is this information-theoretic converse. 

This same converse offers some interesting insights on coding designs for selfish coded caching. While our converse now reveals that such designs, even if they are optimally constructed, would essentially never be able to provide good performance, these designs do pose an exceptionally interesting and challenging coding problem, which we address partially by providing, for a class of demands, achievable schemes whose performance matches the expression of the converse.

\subsection{Paper Outline}

The rest of the paper is organized as follows. The system model is presented in~\Cref{sec: System Model}, where~\Cref{sec: Motivating Example} offers a small motivating example that can help the reader appreciate the dynamics of selfish coded caching. Then, \Cref{sec: Converse Bound} presents the information-theoretic converse, whose proof in \Cref{sec: Proof Converse Bound} is followed by a clarifying example. The proposed selfish coded caching placement is presented in~\Cref{sec: The Scheme for alpha-Demands} and so are the delivery designs for some sets of demands. Additional optimal schemes are presented in~\Cref{sec: Some Optimal Schemes for Circular Demands} for other sets of demands. \Cref{sec: Conclusion} concludes the paper, while some of the proofs are relegated in the appendices.

\subsection{Notation}\label{sec: Notation}  

We denote by $\mathbb{Z}^{+}$ the set of positive integers. For $n \in \mathbb{Z}^{+}$, we define $[n] \coloneqq \{1, 2, \dots, n\}$. If $a, b \in \mathbb{Z}^{+}$ such that $a < b$, then $[a : b] \coloneqq \{a, a + 1, \dots, b - 1, b\}$. For sets we use calligraphic symbols, whereas for vectors we use bold symbols. Given a finite set $\mathcal{A}$, we denote by $|\mathcal{A}|$ its cardinality. We use $\binom{n}{k}$ to denote the binomial coefficient $\frac{n!}{k!(n - k)!}$ and we let $\binom{n}{k} = 0$ whenever $n < 0$, $k < 0$ or $n < k$. We use the $\oplus$ symbol to denote the bitwise XOR operation. For $\bm{u} = (u_{1}, \dots, u_{K})$ being a permutation of the set $[K]$, we use $u \colon [K] \to [K]$ to denote the function which takes as input an element from $[K]$ and outputs its index position in $\bm{u}$.

\section{System Model}\label{sec: System Model}

Similarly to the original scenario in \cite{6763007}, we consider the centralized caching scenario (cf.~\Cref{fig: Centralized Coded Caching Model}) where one central server has access to a library $\mathcal{L}$ containing $N$ files of $B$ bits each. This server is connected to $K$ users through a shared error-free broadcast channel, and each user is equipped with a cache of size $M$ files or, equivalently, $MB$ bits.

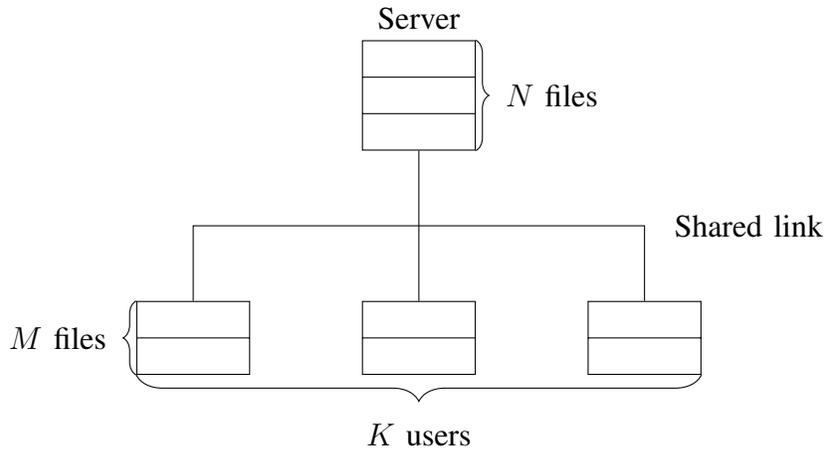
\begin{figure}[!htpb]
  \centering
  \tikzset{server/.style={draw, rectangle, minimum height = 1.5cm, minimum width = 2.5cm, text centered}}
  \tikzset{user/.style={draw, rectangle, minimum height = 0.75cm, minimum width = 2.5cm, text centered}}
  \begin{tikzpicture}
    \node[draw, rectangle split, rectangle split parts = 3, minimum width = 1.5cm](server){};
    \node at (server.north)[above]{Server};
    \draw[decorate, decoration = {brace, amplitude = 5}] (server.north east)--(server.south east) node[midway, right, xshift = 0.25cm]{$N$ files};
    \draw (server.south)--++(0, -1);
    \draw ($(server.south) + (0, -1)$)--++(-3, 0)--++(0, -1) node[draw, rectangle split, rectangle split parts = 2, minimum width = 1.5cm, anchor = north](user1){};
    \draw ($(server.south) + (0, -1)$)--++(0, -1) node[draw, rectangle split, rectangle split parts = 2, minimum width = 1.5cm, anchor = north](user2){};
    \draw ($(server.south) + (0, -1)$)--++(3, 0)--++(0, -1) node[draw, rectangle split, rectangle split parts = 2, minimum width = 1.5cm, anchor = north](user3){};
    \draw[decorate, decoration = {brace, mirror, amplitude = 5}] (user1.north west)--(user1.south west) node[midway, left, xshift = -0.25cm]{$M$ files};
    \draw[decorate, decoration = {brace, mirror, amplitude = 10}] (user1.south west) -- (user3.south east) node[midway, below, yshift = -0.5cm]{$K$ users};
    \node at ($(server.south) + (3, -1)$)[right, xshift = 0.25cm]{Shared link};
  \end{tikzpicture}
  \caption{A server with access to a library of $N$ files is connected through an error-free unit-capacity shared link to $K$ users, each having a cache of size equal to $M$ files.}
  \label{fig: Centralized Coded Caching Model}
\end{figure}

During the placement phase, the server fills the caches of the users according to a caching policy without knowing the future requests. During the delivery phase, when the users simultaneously reveal their demands, the server sends coded messages over the shared link to deliver the missing information to each user. Assuming that in the delivery phase each user demands simultaneously one file, the worst-case communication load $R$ is defined as the total number of transmitted bits, normalized by the file size $B$, that can guarantee delivery of all requested files in the worst-case scenario. The optimal communication load $R^{\star}$ is then formally defined as
\begin{equation}
  R^{\star}(M) \coloneqq \inf \{R : (M, R) \text{ is achievable}\}
\end{equation}
where the tuple $(M, R)$ is said to be \emph{achievable} if there exists a caching-and-delivery scheme which guarantees, for any possible demand, a load $R$.

For the original coded caching scenario in~\cite{6763007} --- where every file is of potential interest to each user --- the load takes the form 
\begin{equation}\label{eqn: MAN Scheme Load}
  R_{\text{MAN}}(t) = \frac{\binom{K}{t + 1}}{\binom{K}{t}} = \frac{K - t}{t + 1} = \frac{K(1 - \gamma)}{K\gamma + 1}, \quad \forall t \in [0 : K]
\end{equation}
where $t \coloneqq KM/N = K\gamma$ is the so-called \emph{cache redundancy} and $\gamma \coloneqq M/N$ is the fraction of the library that each user is able to store. This performance was proven in~\cite{8963629} (see also \cite{8226776}) to be optimal under the assumption of uncoded cache placement. 
The above reveals a speedup factor of $t+1$ over the case of uncoded delivery. This speedup factor is a result of being able to serve \emph{any} $(t + 1)$-tuple of users with a multicast message, which can generally happen if we are able to store bits of each file to \emph{any} possible $t$-tuple of caches. This symmetry will naturally be disrupted once selfish placement is imposed.

\subsection{The Symmetric \texorpdfstring{$(K, \alpha, f)$}{(K, alpha, f)} FDS Structure}

To capture the interplay between coded caching and selfish caching, we propose an FDS structure that allows us to calibrate the degree of separation between the interests of the different users.  To better understand this structure and generally to better understand the concept of an FDS, let us briefly consider a simplified toy example. 

\begin{example}
Consider a downlink scenario with $K = 3$ users and a library $\mathcal{L} = \{A, B, C, D, E, F\}$ of $N = 6$ files\footnote{Such files can be movies, different episodes of a TV show, YouTube videos, etc.}. Let us now assume that user~$1$ is only interested in potentially consuming files from the file demand set $\mathcal{F}_{1} = \{A, B, C, D\}$, user~$2$ only from the set $\mathcal{F}_{2} = \{A, B, E, F\}$, and user~$3$ only from $\mathcal{F}_{3} = \{C, D, E, F\}$. In this setting, each user is interested in a fraction $2/3$ of the library, so for example user~$1$ has no interest in ever consuming the files in $\mathcal{L} \setminus \mathcal{F}_{1} = \{E, F\}$. Similarly, each file is of interest to the same fraction $2/3$ of users, so for example file $A$ is only of interest to user~$1$ and user~$2$. 
\end{example}

For such a setting, we wish to understand the performance of selfish coded caching where each user caches only contents from its own FDS. We proceed with the formal definition of the FDS structure. We note that below an FDS will be defined as a collection of file \emph{classes}, rather than just a collection of files. This allows for more generality and we believe it also better reflects how user preferences are often categorized. 

\begin{definition}[The Symmetric $(K, \alpha, f)$ FDS Structure]
 For $\alpha \in [K]$ and for $f \in \mathbb{Z}^{+}$, the symmetric $(K, \alpha, f)$ FDS structure assumes an $N$-file library $\mathcal{L} = \{\mathcal{W}_{\mathcal{S}} : \mathcal{S} \subseteq [K], |\mathcal{S}| = \alpha \}$ to be a collection of disjoint file classes 
 $\mathcal{W}_{\mathcal{S}} = \{ W_{i, \mathcal{S}} : i \in [f]\}$, with each class $\mathcal{W}_{\mathcal{S}}$ consisting of $f$ different files. In this setting, each user $k \in [K]$ has a File Demand Set
  \begin{equation}
    \mathcal{F}_{k}  = \left\{\mathcal{W}_{\mathcal{S}} : \mathcal{S} \subseteq [K], |\mathcal{S}| = \alpha, k \in \mathcal{S} \right\}
  \end{equation}
which describes the files this user is potentially interested in. 
\end{definition}
As the above says, the library is split into $C = \binom{K}{\alpha}$ disjoint classes of files, corresponding to $N = fC = f\binom{K}{\alpha}$ files in total. The above also says that each user $k$ is interested in its own FDS $\mathcal{F}_{k}$ of $|\mathcal{F}| = |\mathcal{F}_{k}| = f\binom{K - 1}{\alpha - 1}$ files. There are $K$ FDSs, one for each user, and each file class is identified by an $\alpha$-tuple $\mathcal{S}$ that tells us which $\alpha$ users are interested in this class\footnote{In other words, each file belongs to $\alpha$ FDSs. In particular, each file in class $\mathcal{W}_{\mathcal{S}}$ is of interest to the $\alpha$ users in $\mathcal{S}$. Hence, if $\mathcal{S} \ni k$, then the $f$ files in $\mathcal{W}_{\mathcal{S}}$ are in $\mathcal{F}_{k}$ and are thus of interest to user $k$. Finally, under our simplifying assumption that each user has its own FDS, $\alpha$ also describes the number of users interested in any one specific file.}. Finally, we note that $\alpha = 1$ corresponds to the trivial scenario where there is no intersection between the user interests, while $\alpha = K$ corresponds to the traditional unselfish scenario where a common library of $N = f$ files\footnote{In this case we assume $f \geq K$.} is of interest to every user.

In this context, selfish caching places the constraint that each user $k$ can only cache from its own FDS $\mathcal{F}_{k}$. Thus, one key aspect of such selfish caching is that it brings about an increase of the effective normalized cache size for each user. Indeed, whereas in the unselfish scenario each user can cache a fraction 
\begin{equation}
  \gamma = \frac{M}{N} = \frac{t}{K}
\end{equation}
of each file of possible interest, in the selfish scenario this fraction is elevated to a larger 
\begin{equation}
  \gamma_{\alpha} \coloneqq \frac{M}{|\mathcal{F}|} = \frac{t}{\alpha} = \gamma \frac{K}{\alpha}
\end{equation}
which in turn implies a larger local caching gain. 

Deviating from standard notation practices, we will use the double-index notation $W_{f_k, \mathcal{D}_k}$ to denote the file requested by user $k$. Consequently, to describe the entire demand set, we will now be needing two vectors $\bm{d} = (\mathcal{D}_1, \dots, \mathcal{D}_K)$ and $\bm{f} = (f_1, \dots, f_K)$.

The above structure nicely lets us calibrate the fraction 
\begin{equation}
  \frac{|\mathcal{F}|}{N} = \frac{f\binom{K - 1}{\alpha - 1}}{f\binom{K}{\alpha}} = \frac{\alpha}{K}
\end{equation}
of the total library that each user is interested in. The imposed symmetry also yields a fraction $\delta \coloneqq \alpha/K$ of users interested in any one specific file.

\begin{table}[!htbp]
    \centering
    \renewcommand{\arraystretch}{1.3}
    \caption{Important parameters for the symmetric $(K, \alpha, f)$ FDS structure}
    \label{tab: Important Parameters for the FDS Structure}
    \begin{tabular}{lc@{\hspace{1cm}}lc}
    \toprule
    Total FDSs & $K$ & Files per Class & $f$ \\
    \midrule
    Total File Classes & $\binom{K}{\alpha}$ & Total Files & $f\binom{K}{\alpha}$ \\
    \midrule
    File Classes per FDS & $\binom{K - 1}{\alpha - 1}$ & Files per FDS & $f\binom{K - 1}{\alpha - 1}$ \\
    \midrule
    Fraction of Users Interested in a File & $\alpha/K$ & Fraction of Files of Interest to a User & $|\mathcal{F}|/N$ \\
    \bottomrule 
    \end{tabular}
\end{table}

The following two examples can help familiarize the reader with the notation. 

\begin{example}[The Symmetric $(4, 2, 1)$ FDS Structure]
  Let us consider the $(K, \alpha, f) = (4, 2, 1)$ structure which has $C = \binom{K}{\alpha} = 6$ file classes $\mathcal{W}_{12}, \mathcal{W}_{13},\mathcal{W}_{14}, \mathcal{W}_{23}, \mathcal{W}_{24}, \mathcal{W}_{34}$, where\footnote{We will often omit braces and commas when indicating sets, such that for example $W_{\{1, 2\}}$ may be written as $W_{12}$.} each class consists of $f = 1$ file. This corresponds to a library $\mathcal{L} = \{W_{1, 12}, W_{1, 13}, W_{1, 14}, W_{1, 23}, W_{1, 24}, W_{1, 34}\}$ of $N = 6$ files. In the above, $W_{1, 12}$ simply represents the first (and, in this case, the only) file in class $\mathcal{W}_{12}$. The $K = 4$ FDSs take the form
  \begin{align}
    \mathcal{F}_{1} &= \{W_{1, 12}, W_{1, 13}, W_{1, 14}\}\\
    \mathcal{F}_{2} &= \{W_{1, 12}, W_{1, 23}, W_{1, 24}\}\\
    \mathcal{F}_{3} &= \{W_{1, 13}, W_{1, 23}, W_{1, 34}\}\\
    \mathcal{F}_{4} &= \{W_{1, 14}, W_{1, 24}, W_{1, 34}\}
  \end{align}
   where we recall that for each file $W_{1, \mathcal{S}}$, the label $\mathcal{S}$ represents the FDSs the file belongs to. For example, file $W_{1, 23}$ belongs to $\mathcal{F}_{2}$ and $\mathcal{F}_{3}$, and is thus of interest to user~$2$ and user~$3$. Finally we see that each user is interested in a fraction $|\mathcal{F}|/N = 0.5$ of the library, i.e., in $\SI{50}{\percent}$ of the library, and that each file is of interest to a fraction $\delta = \alpha/K = 0.5$ of the users. 
\end{example}

\begin{example}[The Symmetric $(4, 3, 2)$ FDS Structure]\label{ex: Example of FDS Structure}
  Let us consider the $(K, \alpha, f) = (4, 3, 2)$ structure which has $C = \binom{K}{\alpha} = 4$ classes $\mathcal{W}_{123}, \mathcal{W}_{124}, \mathcal{W}_{134}, \mathcal{W}_{234}$ and $N = 8$ files: $W_{1,123}$ and $W_{2,123}$ from class $\mathcal{W}_{123}$, then $W_{1,124}$ and $W_{2,124}$ from class $\mathcal{W}_{124}$, and so on. The $K$ FDSs take the form
  \begin{align}
    \mathcal{F}_{1} &= \{\mathcal{W}_{123}, \mathcal{W}_{124}, \mathcal{W}_{134}\}\\
    \mathcal{F}_{2} &= \{\mathcal{W}_{123}, \mathcal{W}_{124}, \mathcal{W}_{234}\}\\
    \mathcal{F}_{3} &= \{\mathcal{W}_{123}, \mathcal{W}_{134}, \mathcal{W}_{234}\}\\
    \mathcal{F}_{4} &= \{\mathcal{W}_{124}, \mathcal{W}_{134}, \mathcal{W}_{234}\}
  \end{align}
  where we see that each FDS consists of $2\times 3 = 6$ files. For example user~$1$ is interested in files $\mathcal{F}_1 = \{W_{1, 123}, W_{2, 123}, W_{1, 124}, W_{2, 124}, W_{1, 134}, W_{2, 134}\}$, user~$2$ is interested in files $\mathcal{F}_2 = \{W_{1, 123}, W_{2, 123}, W_{1, 124}, W_{2, 124}, W_{1, 234}, W_{2, 234}\}$, and so on. By calculating $|\mathcal{F}|/N = \alpha/K = 3/4$, we can verify that each user is interested in $\SI{75}{\percent}$ of the library, and each file is of interest to $\SI{75}{\percent}$ of the users.
\end{example}

The FDS structure automatically implies restrictions in the set of possible demand vectors. For instance, going back to~\Cref{ex: Example of FDS Structure}, any demand with $\bm{d} = (234, 123, 123, 234)$ is not valid, because $\{1\} \notin \mathcal{D}_1 = \{2, 3, 4\}$, i.e., because file $W_{f_1, 234}$ is not in $\mathcal{F}_1$ and thus would never be demanded by user~$1$. On the other hand, any demand with $\bm{d} = (124, 123, 123, 234)$ is valid because $k \in \mathcal{D}_k$ for each $k \in [K]$. 

The set of valid demands as well as placement constraints that define selfish coded caching are now stated below. 

\begin{definition}[Selfish Coded Caching with Uncoded Placement]\label{def: Uncoded and Selfish Cache Placement Definition}
In selfish coded caching, a demand defined by the vectors $\bm{d} = (\mathcal{D}_1, \dots, \mathcal{D}_K)$ and $\bm{f} = (f_1, \dots, f_K)$ is said to be \emph{valid} if and only if
  \begin{equation}
      k \in \mathcal{D}_k, \quad \forall k \in [K]
  \end{equation}
while a cache placement is \emph{selfish} when it guarantees that a subfile of $W_{i, \mathcal{S}}$ can be cached at user $k$ only if $k \in \mathcal{S}$.
\end{definition}

\subsection{Understanding the Dynamics of Selfish Coded Caching with an Example for the \texorpdfstring{$(K, \alpha, f) = (5, 4, 1)$}{(K, alpha, f) = (5, 4, 1)} Structure}\label{sec: Motivating Example}
Let us consider a small motivating example that can help the reader appreciate the dynamics of symmetrically selfish coded caching. We will first suggest a selfish cache placement scheme that will be justified in~\Cref{sec: The Scheme for alpha-Demands}, and we will then present the delivery and decoding process for a class of valid circular demands. The corresponding load that will be achieved here will in fact be matched by the converse of the next section, thus proving that in our example our delivery is optimal and the converse tight.

We here consider the $(K, \alpha, f) = (5, 4, 1)$ scenario, where each cache is of size $M = 2$ corresponding to the case of $t = 2$. In our scenario there are $C = \binom{K}{\alpha} = 5$ file classes $\mathcal{W}_{1234}, \mathcal{W}_{1235}, \mathcal{W}_{1245}, \mathcal{W}_{1345}, \mathcal{W}_{2345}$, and a total of $N = fC = 5$ library files. For simplicity, we will exploit the fact that $f = 1$ by slightly abusing notation such that, in this early example only, the library of $N = 5$ files will be denoted as $\mathcal{L} = \{W_{1234}, W_{1235}, W_{1245}, W_{1345}, W_{2345}\}$. At this point, the $5$ FDSs take the form
\begin{align}
    \mathcal{F}_{1} &= \{W_{1234}, W_{1235}, W_{1245}, W_{1345}\}\\
    \mathcal{F}_{2} &= \{W_{1234}, W_{1235}, W_{1245}, W_{2345}\}\\
    \mathcal{F}_{3} &= \{W_{1234}, W_{1235}, W_{1345}, W_{2345}\}\\
    \mathcal{F}_{4} &= \{W_{1234}, W_{1245}, W_{1345}, W_{2345}\}\\
    \mathcal{F}_{5} &= \{W_{1235}, W_{1245}, W_{1345}, W_{2345}\}.
\end{align}

\subsubsection{Placement}

The cache placement will follow a selfish adaptation of the MAN scheme. First each file is split into $\binom{\alpha}{t} = \binom{4}{2} = 6$ non-overlapping subfiles as
\begin{align}
    W_{1234} & = \{W_{1234, 12}, W_{1234, 13}, W_{1234, 14}, W_{1234, 23}, W_{1234, 24}, W_{1234, 34}\}\\
    W_{1235} & = \{W_{1235, 12}, W_{1235, 13}, W_{1235, 15}, W_{1235, 23}, W_{1235, 25}, W_{1235, 35}\}\\
    W_{1245} & = \{W_{1245, 12}, W_{1245, 14}, W_{1245, 15}, W_{1245, 24}, W_{1245, 25}, W_{1245, 45}\}\\
    W_{1345} & = \{W_{1345, 13}, W_{1345, 14}, W_{1345, 15}, W_{1345, 34}, W_{1345, 35}, W_{1345, 45}\}\\
    W_{2345} & = \{W_{2345, 23}, W_{2345, 24}, W_{2345, 25}, W_{2345, 34}, W_{2345, 35}, W_{2345, 45}\}
\end{align} 
and then the cache $\mathcal{Z}_{k}$ of each user $k \in [5]$ is filled as 
\begin{equation}
  \mathcal{Z}_{k} = \{W_{\mathcal{S}, \mathcal{T}}: \mathcal{S} \subseteq [5], |\mathcal{S}| = 4, \mathcal{T} \subseteq \mathcal{S}, |\mathcal{T}| = 2, k \in \mathcal{S} \cap \mathcal{T}\}.
\end{equation}
For example, user~$1$ would have to cache parts only from files $\{W_{1234}, W_{1235}, W_{1245}, W_{1345}\}$ in order to abide by the selfish constraint, and then, to abide by the cache size constraint, user~$1$ would cache subfiles labeled by $\{12,13,14\}$. Similarly, user~$2$ would cache only from $\{W_{1234}, W_{1235}, W_{1245}, W_{2345}\}$, and only the subfiles labeled by $\{12,23,24\}$, and so on.

\subsubsection{Delivery}

The delivery takes place as soon as the requests of the users are revealed. Consider the demand $\bm{d}_1 = (1234, 2345, 1345, 1245, 1235)$. A schematic of this demand is given by means of the graph in~\Cref{fig: FDS Request Graph for Circular Demand}. This graph, which we refer to as the \emph{FDS request graph}, is a directed graph where each vertex is a user and where there is an edge from user $k_1$ to user $k_2$ if $W_{\mathcal{D}_{k_1}} \in \mathcal{F}_{k_2}$. This graph represents at a high level, for each given demand vector $\bm{d}$, the interplay between the users' interests.
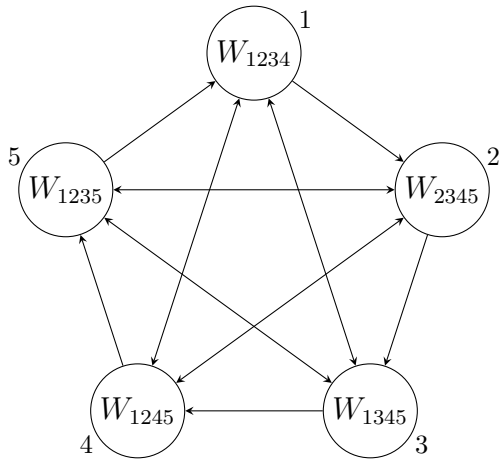
\begin{figure}[!htb]
\centering
\begin{tikzpicture}
  \node[regular polygon, regular polygon sides = 5, inner sep = 1.5cm](G){};
  \node[draw, circle, label = center:$W_{1234}$, minimum size = 1.25cm](1) at (G.corner 1){};
  \node[draw, circle, label = center:$W_{1235}$, minimum size = 1.25cm](5) at (G.corner 2){};
  \node[draw, circle, label = center:$W_{1245}$, minimum size = 1.25cm](4) at (G.corner 3){};
  \node[draw, circle, label = center:$W_{1345}$, minimum size = 1.25cm](3) at (G.corner 4){};
  \node[draw, circle, label = center:$W_{2345}$, minimum size = 1.25cm](2) at (G.corner 5){};
  \node at (1.north east)[font = \small, anchor = west]{$1$};
  \node at (2.north east)[font = \small, anchor = west]{$2$};
  \node at (3.south east)[font = \small, anchor = west]{$3$};
  \node at (4.south west)[font = \small, anchor = east]{$4$};
  \node at (5.north west)[font = \small, anchor = east]{$5$};
  \draw[-stealth](1)--(2); \draw[-stealth](2)--(3); \draw[-stealth](3)--(4); \draw[-stealth](4)--(5); \draw[-stealth](5)--(1);
  \draw[stealth-stealth](1)--(4); \draw[stealth-stealth](2)--(5); \draw[stealth-stealth](3)--(1); \draw[stealth-stealth](4)--(2); \draw[stealth-stealth](5)--(3);
\end{tikzpicture}
\caption{FDS request graph for the $(K, \alpha, f) = (5, 4, 1)$ FDS structure and the demand $\bm{d}_1 = (1234, 2345, 1345, 1245, 1235)$.}
\label{fig: FDS Request Graph for Circular Demand}
\end{figure}

As a consequence of the aforementioned cache placement, each user does not cache (and consequently desires) a total of $\binom{\alpha - 1}{t} = \binom{3}{2} = 3$ subfiles for its demanded file. Hence, given the demand $\bm{d}_1$, the desired subfiles are given as follows.
\begin{itemize}
    \item User~$1$ desires the subfiles $W_{1234, 23}$, $W_{1234, 24}$ and $W_{1234, 34}$.
    \item User~$2$ desires the subfiles $W_{2345, 34}$, $W_{2345, 35}$ and $W_{2345, 45}$.
    \item User~$3$ desires the subfiles $W_{1345, 14}$, $W_{1345, 15}$ and $W_{1345, 45}$.
    \item User~$4$ desires the subfiles $W_{1245, 12}$, $W_{1245, 15}$ and $W_{1245, 25}$.
    \item User~$5$ desires the subfiles $W_{1235, 12}$, $W_{1235, 13}$ and $W_{1235, 23}$.
\end{itemize}
One key aspect for achieving optimality is the utilization of specifically structured linear combinations of multicast messages, where this structure accepts the following interesting interpretation. These linear combinations effectively allow multicast messages to be used not only to deliver desired content to users, but also to deliver undesired content that can be used as side information to ``bridge'' the gaps left by the selfish placement. In essence, each transmission now delivers desired content while also disseminating side information that can be used to create \emph{cliques}. To see this, let us consider the following sequence of XORs
\begin{align}
    X_1 & = W_{1345, 14} \oplus W_{1234, 24} \oplus W_{1245, 12} \\
    X_2 & = W_{2345, 35} \oplus W_{1235, 13} \oplus W_{1345, 15} \\
    X_3 & = W_{1345, 14} \oplus W_{2345, 35} \oplus W_{1234, 23}  \\
    X_4 & = W_{1234, 34} \oplus W_{1245, 15} \\
    X_5 & = W_{2345, 45} \oplus W_{1235, 12} \\
    X_6 & = W_{2345, 34} \oplus W_{1245, 25} \\
    X_7 & = W_{1345, 45} \oplus W_{1235, 23}
\end{align}
transmitted one after the other. Recalling that each file is split into $6$ non-overlapping subfiles, we know that each XOR has size $|X_i| = B/6$ for each $i \in [7]$.

By using its own cache, each user can now decode its own desired content as follows.
\begin{itemize}
    \item User~$1$ can recover its desired subfiles from $X_1$, $X_2 \oplus X_3$ and $X_4$.
    \item User~$2$ can recover its desired subfiles from $X_1 \oplus X_3$, $X_5$ and $X_6$.
    \item User~$3$ can recover its desired subfiles from $X_2$, $X_3$ and $X_7$.
    \item User~$4$ can recover its desired subfiles from $X_1$, $X_4$ and $X_6$.
    \item User~$5$ can recover its desired subfiles from $X_2$, $X_5$ and $X_7$.
\end{itemize}
For example, in the above, user~$1$ needs $W_{2345, 35}$ to correctly decode its desired $W_{1234, 23}$ from $X_3$, whereas user~$2$ needs $W_{1345, 14}$ to correctly decode $W_{2345, 35}$ always from $X_3$. The act of ``passing'' subfiles $W_{2345, 35}$ and $W_{1345, 14}$ to user~$1$ and user~$2$ with $X_2$ and $X_1$, respectively, allows the creation of a clique between user~$1$, user~$2$ and user~$3$. This clique is exploited by creating the XOR $X_3$. This interpretation related to the creation of cliques is a crucial part of the dynamics of the problem that we are considering.

The corresponding communication load is equal to $R(t = 2) = |X|/B = 7/6$, which will be met by the converse. 

\begin{figure}[!htb]
\centering
\begin{tikzpicture}
  \node[regular polygon, regular polygon sides = 5, inner sep = 1.5cm](G){};
  \node[draw, circle, label = center:$W_{1234}$, minimum size = 1.25cm](1) at (G.corner 1){};
  \node[draw, circle, label = center:$W_{1345}$, minimum size = 1.25cm](5) at (G.corner 2){};
  \node[draw, circle, label = center:$W_{1245}$, minimum size = 1.25cm](4) at (G.corner 3){};
  \node[draw, circle, label = center:$W_{1235}$, minimum size = 1.25cm](3) at (G.corner 4){};
  \node[draw, circle, label = center:$W_{2345}$, minimum size = 1.25cm](2) at (G.corner 5){};
  \node at (1.north east)[font = \small, anchor = west]{$1$};
  \node at (2.north east)[font = \small, anchor = west]{$2$};
  \node at (3.south east)[font = \small, anchor = west]{$3$};
  \node at (4.south west)[font = \small, anchor = east]{$4$};
  \node at (5.north west)[font = \small, anchor = east]{$5$};
  \draw[-stealth](1)--(2); \draw[-stealth](2)--(5); \draw[-stealth](5)--(1);
  \draw[stealth-stealth](1)--(3); \draw[stealth-stealth](2)--(3); \draw[stealth-stealth](1)--(4); \draw[stealth-stealth](5)--(4); \draw[stealth-stealth](5)--(3); \draw[stealth-stealth](2)--(4);
\end{tikzpicture}
\caption{FDS request graph for the $(K, \alpha, f) = (5, 4, 1)$ FDS structure and the demand $\bm{d} = (1234, 2345, 1235, 1245, 1345)$.}
\label{fig: FDS Request Graph for Non-Circular Demand}
\end{figure}

Consider now another demand $\bm{d}_2 = (1234, 2345, 1235, 1245, 1345)$ with its corresponding FDS request graph shown in~\Cref{fig: FDS Request Graph for Non-Circular Demand}. Since the graphs in~\Cref{fig: FDS Request Graph for Circular Demand} and in~\Cref{fig: FDS Request Graph for Non-Circular Demand} are non-isomorphic\footnote{This can be concluded by noticing that the graph in~\Cref{fig: FDS Request Graph for Circular Demand} contains $5$ bidirectional edges, whereas the graph in~\Cref{fig: FDS Request Graph for Non-Circular Demand} has $6$ bidirectional edges.}, the demand $\bm{d}_2$ accepts a different delivery solution\footnote{Having two non-isomorphic problems here implies that the delivery for the second problem cannot be derived from that of the first problem by a simple relabeling of the users.} than that for demand $\bm{d}_1$. Such phenomenon does not happen in the standard coded caching scenario, where indeed each demand would result in the same FDS request graph (cf.~\Cref{fig: FDS Request Graph for Any Demand in Standard MAN}), which is always complete\footnote{A complete graph is a graph where every node is connected to every other node.}. In such an unselfish scenario where each file is assumed to be of interest to all users, every user in the FDS request graph is connected to every other user, independently of the requested files. Hence, in the unselfish scenario, having a fixed FDS request graph for every demand allows for an identical delivery procedure for any demand. This seems to be a crucial differentiating aspect between selfish and unselfish coded caching.

\begin{figure}[!htb]
\centering
\begin{tikzpicture}
  \node[regular polygon, regular polygon sides = 5, inner sep = 1.5cm](G){};
  \node[draw, circle, label = center:$W_{f_1}$, minimum size = 1.25cm](1) at (G.corner 1){};
  \node[draw, circle, label = center:$W_{f_5}$, minimum size = 1.25cm](5) at (G.corner 2){};
  \node[draw, circle, label = center:$W_{f_4}$, minimum size = 1.25cm](4) at (G.corner 3){};
  \node[draw, circle, label = center:$W_{f_3}$, minimum size = 1.25cm](3) at (G.corner 4){};
  \node[draw, circle, label = center:$W_{f_2}$, minimum size = 1.25cm](2) at (G.corner 5){};
  \node at (1.north east)[font = \small, anchor = west]{$1$};
  \node at (2.north east)[font = \small, anchor = west]{$2$};
  \node at (3.south east)[font = \small, anchor = west]{$3$};
  \node at (4.south west)[font = \small, anchor = east]{$4$};
  \node at (5.north west)[font = \small, anchor = east]{$5$};
  \draw[stealth-stealth](1)--(2); \draw[stealth-stealth](1)--(3); \draw[stealth-stealth](1)--(4); \draw[stealth-stealth](1)--(5);
  \draw[stealth-stealth](2)--(3); \draw[stealth-stealth](2)--(4); \draw[stealth-stealth](2)--(5);
  \draw[stealth-stealth](3)--(4); \draw[stealth-stealth](3)--(5);
  \draw[stealth-stealth](4)--(5);
\end{tikzpicture}
\caption{FDS request graph for any demand in the standard (unselfish) MAN scenario with $K = 5$ users and $N = 5$ files labeled as $W_i$ with $i \in [5]$. In this case the demand is identified by the vector $\bm{f} = (f_1, f_2, f_3, f_4, f_5)$, where user $k \in [5]$ requests file $W_{f_k}$. This graph is complete. Hence, here the ability to create cliques of subfiles is only limited by $t$, and is not affected at all by the specific demand.}
\label{fig: FDS Request Graph for Any Demand in Standard MAN}
\end{figure}
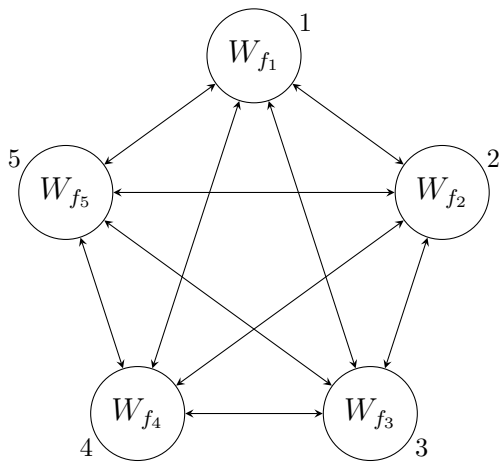

\section{Converse Bound for Selfish Coded Caching with Uncoded Placement}\label{sec: Converse Bound}

Let us recall that each user is interested in its own FDS, and that each FDS only represents a fraction $|\mathcal{F}|/N$ of the library. In the general unselfish scenario, a portion $(1 - |\mathcal{F}|/N)$ of each user's cache would be filled with content that would never be requested by that user. Such a non-selfish scheme would relinquish local caching gain for the benefit of being able to encode across all combinations of users. Under the basic clique-based approach in the MAN scheme, we are presented with a trade-off between local caching gain and coding gain, where the latter seems to be more desirable. Are there though other coding techniques that manage to harvest an abundance of coding opportunities, which are usually associated to the standard coded caching approach, exploiting the existence of a more targeted set of demands, while capitalizing on the increased local caching gain brought about by a selfish variant? If not, then what is the amount of coding gain that can be harvested while maintaining selfish caching? These are the questions addressed by our information-theoretic converse that lower bounds the optimal worst-case load assuming uncoded and selfish cache placement.

\subsection{Theorem Statement}

The converse bound employs the index coding techniques of~\cite{8963629} that proved the optimality of the MAN scheme under the constraint of uncoded cache placement. Our main challenge will be to account for the presence of different profiles of interest, adapting consequently the index coding approach to reflect the $(K, \alpha, f)$ FDS structure proposed in the previous section. The converse bound presented here shows that adding the selfish cache placement constraint implies a higher optimal communication load compared to the unselfish scenario. The result is stated in the following theorem. We recall that $\gamma_\alpha =  \gamma K/\alpha$ is the effective normalized cache size, and that $t = K\gamma = \alpha\gamma_\alpha$ is the cache redundancy. We also recall that $K\gamma+1$ is the optimal coding gain for the unselfish scenario.

\begin{theorem}[Converse Bound for Selfish Coded Caching under Uncoded Prefetching]\label{thm: Lower Bound for the FDS Structure}
  Under the assumption of uncoded and selfish cache placement, and given the $(K, \alpha, f)$ FDS structure, the optimal worst-case communication load $R^{\star}$ is lower bounded by $R_{\text{LB}}$ which is a piece-wise linear curve with corner points
  \begin{equation}
    (M, R_{\text{LB}}) = \left(t\frac{N}{K}, \frac{\binom{\alpha}{t + 1} + (K - \alpha)\binom{\alpha - 1}{t}}{\binom{\alpha}{t}}\right), \quad \forall t \in [0 : \alpha]
  \end{equation}
corresponding to 
  \begin{equation}
    R_{\text{LB}} = \frac{K(1 - \gamma_\alpha)}{K\gamma+1}\Big[(K - \alpha)\gamma+1\Big].
  \end{equation}
\end{theorem}

\begin{proof}
  We provide the proof of the converse in~\Cref{sec: Proof of the Lower Bound}. In~\Cref{sec: An Exhaustive Example for the Converse Bound} we also present an example that aims to help the reader better understand the construction of the outer bound.
\end{proof}

\subsection{Comments on the Converse Bound}

The bound reveals some interesting insights. Before discussing these insights, let us quickly recall that, in our scenario, the integer value $t$ is upper bounded by $\alpha$, since any $t \geq \alpha$ would imply zero communication load. 

\subsubsection{Comparison with MAN}

The following compares, for any $f$, the optimal load $R^{\star}(t)$ of selfish coded caching with that of the unselfish (MAN) scenario\footnote{The comparison between the selfish and unselfish scenarios is made easy by the fact that the $t$ values (i.e., the integer points corresponding to the memory-axis of the memory-load trade-off) in the two scenarios coincide.}.

\begin{corollary}\label{cor: Comparison between Converse Bound and MAN Scheme}
Given the symmetric $(K, \alpha, f)$ FDS structure and $\alpha \in [K - 1]$, the converse reveals that 
\begin{align}
    \frac{R^{\star}(t)}{R_{\text{MAN}}(t)} & \geq 1,  \quad \forall t \in [0 : \alpha - 1] \\
    \frac{R^{\star}(t)}{R_{\text{MAN}}(t)} & > 1,  \quad \forall t \in (0 : \alpha - 1) 
\end{align}
which says that in the non-trivial range $t \in [0 : \alpha - 1]$, selfish coded caching is not better than unselfish coded caching, while in the non-extremal points of $t$ and under uncoded placement optimal unselfish coded caching strictly outperforms any implementation of selfish coded caching. When $\alpha = K$ and $f \geq K$, the converse expression naturally matches that of unselfish coded caching.
\end{corollary}

\begin{proof}
  The proof can be found in~\refappendix{app: Proof of Comparison with MAN}, while a graphical comparison can be found in~\Cref{fig: Comparison between the Converse Bound and the MAN Scheme}.
\end{proof}

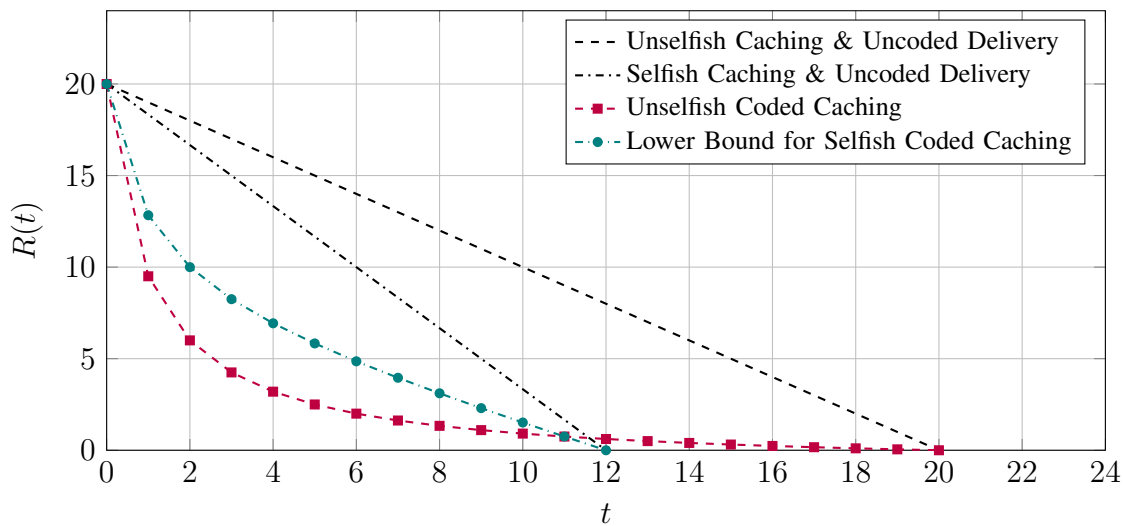
\begin{figure}[!htb]
  \centering
  \def\K{20}
  \def\a{12}
  \begin{tikzpicture}[ 
    declare function = {% no spaces allowed between variables accepted within the function definition of binom
      binom(\n,\k) = (\n >= \k)*(\n!)/(\k!*(\n-\k)!);
      MAN(\K, \t) = binom(\K, \t + 1)/binom(\K, \t);
      LB(\K, \a, \t) = (binom(\a, \t + 1) + (\K - \a)*binom(\a - 1, \t))/binom(\a, \t);
    }]
    \begin{axis}[xlabel = $t$, ylabel = $R(t)$, grid = major, enlargelimits = {value = 0.2, upper}, legend cell align = {left}, legend style = {font = \small},
      width = 0.9\linewidth, height = 0.45\textwidth]
      \addplot+[variable = t, samples at = {0, 1, ..., \K}, mark = none, thick, black, dashed]{\K - t};
      \addlegendentry{Unselfish Caching \& Uncoded Delivery};
      \addplot+[variable = t, samples at = {0, 1, ..., \a}, mark = none, thick, black, dash dot]{\K - \K*t/\a};
      \addlegendentry{Selfish Caching \& Uncoded Delivery};
      \addplot+[variable = t, samples at = {0, 1, ..., \K}, mark = square*, mark size = 1.5, mark options = {solid}, dashed, thick, purple] {MAN(\K, t)};
      \addlegendentry{Unselfish Coded Caching};
      \addplot+[variable = t, samples at = {0, 1, ..., \a}, mark = *, mark size = 1.5, mark options = {solid}, dash dot, thick, teal] {LB(\K, \a, t)};
      \addlegendentry{Lower Bound for Selfish Coded Caching};
    \end{axis}
  \end{tikzpicture}
  \caption{Comparison between selfish and unselfish caching for the $(20, 12, f)$ FDS structure.}
  \label{fig: Comparison between the Converse Bound and the MAN Scheme}
\end{figure}

\subsubsection{Selfish Local Caching Gain and Coding Gain}

We recall that, in the presence of a relatively small $\alpha$, selfish caching implies a sizeable increase in the effective normalized cache size $\gamma_{\alpha} = \gamma K/\alpha$, which in turn implies a much larger local caching gain. 

On the other hand, the converse reveals that a smaller $\alpha$ implies a substantial reduction in the coding gain offered by selfish coded caching. To compare coding gains, we first recall that the coding gain in the original unselfish scenario takes the form
\begin{equation}
    \frac{R_{\text{U}}}{R_{\text{MAN}}} = K\gamma + 1,
\end{equation}
where $R_{\text{U}} = K(1 - \gamma)$ is the load for uncoded delivery. As previously stated, this coding gain $K\gamma + 1$ describes the speedup factor over the uncoded case. To reflect this same speedup in the selfish scenario, we must consider that the corresponding load in the uncoded scenario takes the form $R_{\text{U}, \text{selfish}} = K(1 - \gamma_{\alpha})$. With this in place, the converse reveals that the optimal coding gain of selfish coded caching is upper bounded as
\begin{equation}
  G^\star \leq \frac{R_{\text{U}, \text{selfish}}}{R_{\text{LB}}} = \frac{K\gamma + 1}{(K - \alpha)\gamma +1}
\end{equation}
where the value
\begin{equation}D \coloneqq (K - \alpha)\gamma +1\end{equation} 
represents the guaranteed deterioration in the coding gain when we choose to cache selfishly. Indeed, if we consider the non-trivial range $\alpha \in [2 : K - 1]$, we have $D > 1$ and consequently $G < K\gamma + 1$ for $\gamma > 0$. We can see that --- for fixed $K$ and $\gamma$ --- this deterioration $D$ increases with decreasing $\alpha$, reflecting the fact that the closer the $(K, \alpha, f)$ FDS structure is to the standard MAN scenario, the smaller this deterioration $D$ is.

An important observation though is that the coding gain of selfish coded caching does not scale with $K$. This is described in the following corollary.

\begin{corollary}\label{cor: Bound on the Coding Gain}
    For any fixed ratio $\delta = \alpha/K < 1$ the coding gain of selfish caching does not scale as $K$ increases, and it is instead bounded as 
    \begin{equation}
        G^\star < \frac{1}{1 - \delta}.
    \end{equation}
\end{corollary}
\begin{proof}
  The proof can be found in~\refappendix{app: Proof of the Bound on the Coding Gain}.
\end{proof}

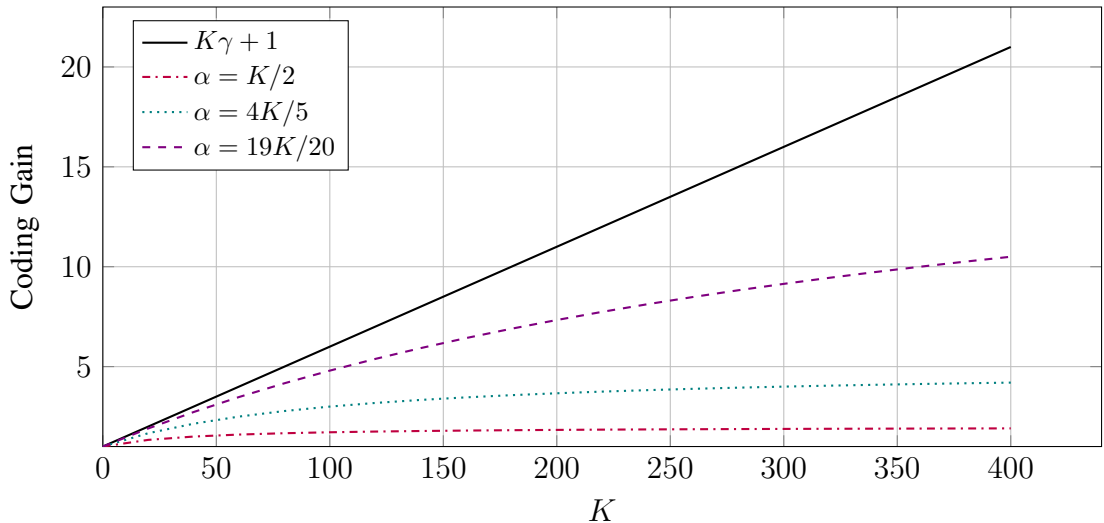
\begin{figure}[!htb]
  \centering
  \def\gam{1/20}
  \begin{tikzpicture}[ 
    declare function = {% no spaces allowed between variables accepted within the function definition of binom
      MAN_G(\K,\gamma) = 1 + \K*\gamma;
      G(\K,\a,\gamma) = (1 + \K*\gamma)/(1 + \gamma*(\K - \a));
    }]
    \begin{axis}[xlabel = $K$, ylabel = {Coding Gain}, grid = major, enlargelimits = {value = 0.1, upper}, legend pos = {north west}, legend cell align = {left}, legend style = {font = \small},
      width = 0.9\linewidth, height = 0.45\textwidth]
      \addplot+[variable = K, samples at = {0, 20, 40, ..., 400}, mark = none, thick, black]{MAN_G(K,\gam)};
      \addlegendentry{$K{\gamma} + 1$};
      \addplot+[variable = K, samples at = {0, 20, 40, ..., 400}, mark = none, thick, purple, dash dot]{G(K, K/2,\gam)};
      \addlegendentry{$\alpha = K/2$};
      \addplot+[variable = K, samples at = {0, 20, 40, ..., 400}, mark = none, thick, teal, dotted]{G(K, 4/5*K,\gam)};
      \addlegendentry{$\alpha = 4K/5$};
      \addplot+[variable = K, samples at = {0, 20, 40, ..., 400}, mark = none, thick, violet, dashed]{G(K, 19/20*K,\gam)};
      \addlegendentry{$\alpha = 19K/20$};
    \end{axis}
  \end{tikzpicture}
  \caption{Plot of different coding gains $G$ for varying values of $K > 20$ and $\alpha$ for the $(K, \alpha, f)$ FDS structure when $\gamma = 1/20$.}
  \label{fig: Comparison of Coding Gains}
\end{figure}

We can see in \Cref{fig: Comparison of Coding Gains} the comparison between different coding gains for varying values of $K$ and $\alpha$ when the normalized cache size $\gamma$ is fixed. As mentioned, smaller values of $\alpha$ correspond to much smaller coding gains. As stated in \Cref{cor: Bound on the Coding Gain}, each curve is upper bounded\footnote{When $\alpha = 1$ it naturally holds that $G^\star = 1$, since in such case uncoded delivery is optimal.} by $1/(1 - \delta)$.

\begin{remark}
At this point, we ought to point out that our choice of having a fully symmetric FDS structure may indeed be an overly penalizing condition. However, this choice exemplifies the mechanisms and effects that come about when selfishness is considered. This same choice nicely offers a crisp method for calibrating the intersection between the interests of the different users, taking us from a scenario where the intersection is minimal, to scenarios ever closer to the original MAN setting where the interests are identical.
\end{remark}

\section{Proof of the Main Information-Theoretic Converse in~\texorpdfstring{\Cref{thm: Lower Bound for the FDS Structure}}{Theorem \ref{thm: Lower Bound for the FDS Structure}}} \label{sec: Proof Converse Bound}

The derivation of the converse makes extensive use of the connection between caching and index coding. This connection was made in~\cite{6763007} and was successfully used in \cite{8963629} to derive the optimal performance of the unselfish scenario. 

We quickly recall that an index coding problem \cite{5714242, CIT-094, 7839236, 8272005} consists of a server wishing to deliver $N'$ independent messages to $K'$ users via a basic bottleneck link. Each user $k \in [K']$ has its own \emph{desired message set} $\mathcal{M}_{k} \subseteq [N']$, and has knowledge of its own \emph{side information set} $\mathcal{A}_{k} \subseteq [N']$. Let $M_i$ be the message $i$ in the set $[N']$. Then, the index coding problem is typically described by its \emph{side information graph} in the form of a directed graph, where each vertex is a message and where there is an edge from $M_i$ to $M_j$ if $M_i$ is in the side information set of the user requesting $M_j$. 
The derivation of our converse will use the following well-known result from~\cite[Corollary 1]{6620369}. 
\begin{lemma}[{\cite[Corollary 1]{6620369}}]\label{thm: Acyclic Subgraph Converse Bound}
  In an index coding problem with $N'$ messages $M_{i}$ for $i \in [N']$, the minimum number of transmitted bits $\rho$ is bounded as
  \begin{equation}
   \rho \geq \sum_{i \in \mathcal{J}} |M_{i}| 
  \end{equation}
for any acyclic subgraph $\mathcal{J}$ of the problem's side information graph.
\end{lemma}

Before proceeding with the main proof, we also recall that under the $(K, \alpha, f)$ FDS structure we have $\mathcal{L} = \{\mathcal{W}_{\mathcal{S}} : \mathcal{S} \subseteq [K], |\mathcal{S}| = \alpha \}$, where $\mathcal{W}_{\mathcal{S}} = \{W_{i, \mathcal{S}} : i \in [f]\}$ is a class of files. We further recall that there are $C = \binom{K}{\alpha}$ classes of files and $N = fC = f\binom{K}{\alpha}$ files. Additionally, we recall that the FDS of each user $k \in [K]$ is given by
\begin{equation}
  \mathcal{F}_{k} = \{\mathcal{W}_{S} : \mathcal{S} \subseteq [K], |\mathcal{S}| = \alpha, k \in \mathcal{S}\}
\end{equation}
that each file has size $B$ bits, and that each user is equipped with a cache of size $MB$ bits. Finally, let us remember that we are interested in the non-trivial range\footnote{When $\alpha = 1$ the proof is trivial, since for such case we have only two integer points corresponding to $t \in \{0, 1\}$: when $t = 0$ the load is equal to $K$, and when $t = 1$ each user has enough memory to cache entirely its own FDS and the load is equal to $0$. Then, the case $\alpha = K$ and $f \geq K$ is equivalent to the standard (unselfish) MAN scenario, which was already considered in~\cite{8963629}.} $\alpha \in [2 : K - 1]$ and in the range $M \in \left[0 : f\binom{K - 1}{\alpha - 1}\right]$ simply because having $M = |\mathcal{F}_k| = f\binom{K - 1}{\alpha - 1}$ implies $R^\star\left(|\mathcal{F}|\right) = 0$ as a consequence of each user being able to store the entirety of its FDS.

\subsection{Main Proof}\label{sec: Proof of the Lower Bound}

The first step toward the converse consists of splitting each file in a generic manner into a maximum of $2^{|\mathcal{S}|} = 2^{\alpha}$ disjoint subfiles as
\begin{equation}
  W_{i, \mathcal{S}} = \{W_{i, \mathcal{S}, \mathcal{T}} : \mathcal{T} \subseteq \mathcal{S}\}, \quad \forall \mathcal{S} \subseteq [K] : |\mathcal{S}| = \alpha, \quad \forall i \in [f]
\end{equation}
where $W_{i, \mathcal{S}, \mathcal{T}}$ is the subfile of $W_{i, \mathcal{S}}$ cached exactly and only by users in $\mathcal{T}$. As already mentioned in~\Cref{def: Uncoded and Selfish Cache Placement Definition}, splitting each file in this way satisfies the selfish cache placement constraint, since $\mathcal{T} \subseteq \mathcal{S}$ and $\mathcal{W}_{\mathcal{S}} \in \mathcal{F}_{k}$ for each $k \in \mathcal{S}$.

\subsubsection{Constructing the Index Coding Problem}

We now make the connection to index coding and we consider the index coding problem with $K' = K$ users and $N' = K2^{\alpha - 1} $ messages, such that for any demand, identified by the vectors $\bm{d} = (\mathcal{D}_{1}, \dots, \mathcal{D}_{K})$ and $\bm{f} = (f_{1}, \dots, f_{K})$, the desired message set and the side information set are respectively given by
\begin{align}
  \mathcal{M}_{k} & = \{W_{f_{k}, \mathcal{D}_{k}, \mathcal{T}} : \mathcal{T} \subseteq \mathcal{D}_{k}, k \notin \mathcal{T} \} \\
  \mathcal{A}_{k} & = \{W_{i, \mathcal{S}, \mathcal{T}} : i \in [f], \mathcal{S} \subseteq [K], |\mathcal{S}| = \alpha, \mathcal{T} \subseteq \mathcal{S}, k \in \mathcal{S} \cap \mathcal{T}\}
\end{align}
for each user $k \in [K]$. For this setting the side information graph takes the form of a directed graph where each subfile represents a vertex, and where there is a connection from (the node corresponding to) $W_{f_{k_{1}}, \mathcal{D}_{k_{1}}, \mathcal{T}_{1}}$ to $W_{f_{k_{2}}, \mathcal{D}_{k_{2}}, \mathcal{T}_{2}}$ if and only if $W_{f_{k_{1}}, \mathcal{D}_{k_{1}}, \mathcal{T}_{1}} \in \mathcal{A}_{k_{2}}$, i.e., if and only if $k_{2} \in \mathcal{T}_{1}$. To apply \Cref{thm: Acyclic Subgraph Converse Bound}, we are interested in acyclic sets of vertices $\mathcal{J}$ in such side information graph. In the spirit of \cite{8963629}, we know that the set
\begin{equation}\label{eqn: Acyclic Set of Vertices}
  \bigcup_{k \in [K]} \bigcup_{\mathcal{T} \subseteq [K] \setminus \{u_{1}, \dots, u_{k}\} \cap \mathcal{D}_{u_{k}}} \left\{W_{f_{u_{k}}, \mathcal{D}_{u_{k}}, \mathcal{T}}\right\}
\end{equation}
does not contain any directed cycle\footnote{Notice that \cite[Lemma 1]{8963629} considers in fact $\mathcal{T} \subseteq [K] \setminus \{u_{1}, \dots, u_{k}\}$ and not $\mathcal{T} \subseteq [K] \setminus \{u_{1}, \dots, u_{k}\} \cap \mathcal{D}_{u_{k}}$. However, the latter is a subset of the former, thus the lemma still holds.} for any demand $(\bm{d}, \bm{f})$ and any vector $\bm{u}$, where $\bm{u} = (u_1, \dots, u_K)$ is a permutation of the users in $[K]$. Consequently, applying \Cref{thm: Acyclic Subgraph Converse Bound} yields the following lower bound
\begin{equation}\label{eqn: Index Coding Lower Bound}
  BR^{\star} \geq \sum_{k \in [K]} \sum_{\mathcal{T} \subseteq [K] \setminus \{u_{1}, \dots, u_{k}\} \cap \mathcal{D}_{u_{k}}} \left| W_{f_{u_{k}}, \mathcal{D}_{u_{k}}, \mathcal{T}} \right|.
\end{equation}

\subsubsection{Selection of Distinct Demands}

Now we wish to create several lower bounds as the one in \eqref{eqn: Index Coding Lower Bound} considering different user permutations $\bm{u}$, and considering a subset of user demands --- each determined by the tuple $(\bm{d}, \bm{f})$ with $\bm{d} = (\mathcal{D}_1, \dots, \mathcal{D}_K)$ and $\bm{f} = (f_1, \dots, f_K)$. Our aim is to eventually average these bounds in order to obtain a useful lower bound on the optimal communication load. For $\mathcal{C}$ being the set of properly selected demands described further below and $\mathcal{U}_{(\bm{d}, \bm{f})}$ being the set of selected user permutations for each demand $(\bm{d}, \bm{f})$ in $\mathcal{C}$, we seek to characterize the expression given by
\begin{equation}\label{eqn: Complete Lower Bound}
    BR^\star \sum_{(\bm{d}, \bm{f}) \in \mathcal{C}} |\mathcal{U}_{(\bm{d}, \bm{f})}| \geq \sum_{(\bm{d}, \bm{f}) \in \mathcal{C}}  \sum_{\bm{u} \in \mathcal{U}_{(\bm{d}, \bm{f})}}\sum_{k \in [K]} \sum_{\mathcal{T} \subseteq [K] \setminus \{u_{1}, \dots, u_{k}\} \cap \mathcal{D}_{u_{k}}} \left| W_{f_{u_{k}}, \mathcal{D}_{u_{k}}, \mathcal{T}} \right|.
\end{equation}
Notice that the goal of carefully selecting the demand set $\mathcal{C}$ and the permutation set $\mathcal{U}_{(\bm{d}, \bm{f})}$ is twofold. The first is to provide the symmetry that will allow us to simplify \eqref{eqn: Complete Lower Bound} into a meaningful expression, and the second is to force the bound to be as tight as possible. 

Toward this, we proceed to select $\mathcal{C}$ to contain circular demands, as these are defined as follows.

\begin{definition}[Circular Demands]\label{def: Circular Demands}
  A demand defined by the vectors $\bm{d} = (\mathcal{D}_{1}, \dots, \mathcal{D}_{K})$ and $\bm{f} = (f_{1}, \dots, f_{K})$ is said to be a \emph{circular demand} if there exists a permutation $\hat{\bm{u}} = (\hat{u}_{1}, \dots, \hat{u}_{K})$ of the set of users $[K]$ such that\footnote{We imply $i \bmod K$ whenever $i > K$.}
  \begin{equation}\label{eqn: Circular Demand Property}
    W_{f_{\hat{u}_{k}}, \mathcal{D}_{\hat{u}_{k}}} \in \bigcap_{i = k + 1}^{k + \alpha - 1}\{\mathcal{F}_{\hat{u}_{i}}\}
  \end{equation}
  for each $k \in [K]$. This simply means that the demand reflects a circular pattern if $\cup_{i = k + 1}^{k + \alpha - 1} \{\hat{u}_{i}\} = \mathcal{D}_{\hat{u}_{k}} \setminus \{\hat{u}_{k}\}$ for each $k \in [K]$. 
\end{definition}

This new definition allows us to describe the sets $\mathcal{C}$ and $\mathcal{U}_{(\bm{d}, \bm{f})}$ as
\begin{align}
    \mathcal{C} & \coloneqq \{(\bm{d}, \bm{f}) : \text{the demand $(\bm{d}, \bm{f})$ is circular} \} \\
    \mathcal{U}_{(\bm{d}, \bm{f})} & \coloneqq \{\text{$K$ circular shifts of the vector $\hat{\bm{u}}$ associated to a given $(\bm{d}, \bm{f}) \in \mathcal{C}$}\}.
\end{align}
These sets will generally yield larger acyclic subgraphs\footnote{This is based on the following observation. Throughout various examples, such demands generally yielded the largest bounds compared to other classes of demands.} in~\eqref{eqn: Acyclic Set of Vertices} that can be used to increase the right-hand side in~\eqref{eqn: Index Coding Lower Bound}, and to provide a better lower bound on $R^\star$.

\subsubsection{Counting the Selected Demands}

Our goal now is to simplify~\eqref{eqn: Complete Lower Bound} into a more meaningful expression. We start by counting how many circular demands there are. To do so, we observe that there is a one-to-one correspondence between one circular demand $(\bm{d}, \bm{f}) \in \mathcal{C}$ and the corresponding set $\mathcal{U}_{(\bm{d}, \bm{f})}$ of permutations of users. This is easy to see, and the intuition is as follows. A demand is said to be circular if there exists a particular ordering of users such that the property in \eqref{eqn: Circular Demand Property} is satisfied. Such ordering is described by the vector $\hat{\bm{u}}$ and is clearly preserved under any circular shift of such vector. Consequently evaluating $|\mathcal{C}|$ is equivalent to counting the vectors $\hat{\bm{u}}$, since each of them corresponds to a distinct circular demand.

Let us focus on user $k \in [K]$. Counting the total number of circular demands where user $k$ requests the file $W_{f_k, \mathcal{D}_k}$ is equivalent to counting the total number of vectors $\hat{\bm{u}}$ such that $\cup_{i = k + 1}^{ k + \alpha - 1}\{\hat{u}_i\} = \mathcal{D}_k \setminus \{k\}$ and $\cup_{i = k + \alpha}^{K + k - 1} \{\hat{u}_i\} = [K] \setminus \mathcal{D}_k$. Recalling that we imply $i \bmod K$ whenever $i > K$, we see that there are $(\alpha - 1)!(K - \alpha)!f^{K - 1}$ such vectors. Then, if we recall that user $k$ can request a total of $f\binom{K - 1}{\alpha - 1}$ files, we see that the total number of circular demands is equal to
\begin{equation}
    |\mathcal{C}| = f\binom{K - 1}{\alpha - 1}(\alpha - 1)!(K - \alpha)!f^{K - 1} = f^K(K - 1)!.
\end{equation}
Furthermore, since $|\mathcal{U}_{(\bm{d}, \bm{f})}| = K$ for each circular demand, we can see that there is a total of $\sum_{(\bm{d}, \bm{f}) \in \mathcal{C}}|\mathcal{U}_{(\bm{d}, \bm{f})}| = f^K K!$ lower bounds --- as the one in \eqref{eqn: Index Coding Lower Bound} --- created for the expression in \eqref{eqn: Complete Lower Bound}.

\subsubsection{Constructing the Optimization Problem}

We will seek to simplify the expression in \eqref{eqn: Complete Lower Bound}, and then to minimize the new simplified expression, in order to lower bound the optimal worst-case load $R^\star$. Toward simplifying, we first count how many times each subfile $W_{i, \mathcal{S}, \mathcal{T}}$ appears in \eqref{eqn: Complete Lower Bound}, where $i \in [f]$, $\mathcal{S} \subseteq [K]$ with $|\mathcal{S}| = \alpha$, $\mathcal{T} \subseteq \mathcal{S}$ and $|\mathcal{T}| = t$. For this purpose, we make use of the following lemma.

\begin{lemma}\label{lem: Circular Shift Lemma}
  Let $\hat{\bm{u}} = (\hat{u}_{1}, \dots, \hat{u}_{K})$ be a permutation of the elements in $[K]$, let $\mathcal{U}$ be the set composed of the $K$ circular shifts of the vector $\hat{\bm{u}}$, and let $k_1, k_2 \in [K]$ such that $k_1 \neq k_2$. Consider
  \begin{equation}
      \ell \coloneqq \left|\bigcup_{i = \hat{u}(k_1) + 1}^{\hat{u}(k_2)} \{\hat{u}_i\}\right|
  \end{equation}
  where we assume $i \bmod K$ whenever $i > K$. Then, there is a total of $(K - \ell)$ vectors $\bm{u} \in \mathcal{U}$ such that $k_1$ appears before $k_2$ in the vector $\bm{u}$.
\end{lemma}

\begin{proof}
  The proof is reported in~\refappendix{app: Proof of the Circular Shift Lemma}.
\end{proof}

Let us focus on subfile $W_{i, \mathcal{S}, \mathcal{T}}$. We start by considering all circular demands $(\bm{d}, \bm{f}) \in \mathcal{C}$ such that $f_{k_1} = i$ and $\mathcal{D}_{k_1} = \mathcal{S}$ for some $k_1 \in \mathcal{S}$ and $\mathcal{T} \subseteq \mathcal{S} \setminus \{k_1\}$ with $|\mathcal{T}| = t$. For each of these circular demands we have a vector $\hat{\bm{u}}$ of ordered users and we select as user permutations the vectors in the set $\mathcal{U}_{(\bm{d}, \bm{f})}$. Considering how the acyclic set of vertices in \eqref{eqn: Acyclic Set of Vertices} is built, it is clear that $W_{i, \mathcal{S}, \mathcal{T}}$ appears in \eqref{eqn: Complete Lower Bound} whenever \emph{all} the elements in $\mathcal{T}$ appear after $k_1$ in $\bm{u} \in \mathcal{U}_{(\bm{d}, \bm{f})}$. By \Cref{lem: Circular Shift Lemma}, we know that this happens a total of $(K - \ell)$ times, where
\begin{equation}
    \ell = \max_{j \in \mathcal{T}} \left| \bigcup_{i = \hat{u}(k_1) + 1}^{\hat{u}(j)} \{\hat{u}_i\} \right|.
\end{equation}
Notice that the above maximization is required since our aim is to count --- for any given $\hat{\bm{u}}$, and thus for a given circular demand --- how many times \emph{all} the elements in $\mathcal{T}$ appear after\footnote{Indeed, recalling that we build acyclic subgraphs as in~\eqref{eqn: Acyclic Set of Vertices}, the subfile $W_{f_{k_1}, \mathcal{D}_{k_1}, \mathcal{T}}$ appears in~\eqref{eqn: Complete Lower Bound} only when \emph{all} elements in $\mathcal{T}$ are after $k_1$ in the vector $\bm{u} \in \mathcal{U}_{(\bm{d}, \bm{f})}$ and the maximization is needed to count in how many of such user permutations it happens to have \emph{all} the elements in $\mathcal{T}$ after $k_1$.} $k_1$ when considering all the users permutations in $\mathcal{U}_{(\bm{d}, \bm{f})}$.

Recalling that $|\mathcal{T}| = t$, we observe that $\ell \in [t : \alpha - 1]$. To see this, we note that the minimum value of $\ell$ is $t$ when $k_1$ and all the elements in $\mathcal{T}$ are in consecutive positions in the vector $\hat{\bm{u}}$. Additionally, we also note that the maximum value of $\ell$ is $(\alpha - 1)$, because $\mathcal{T} \subseteq \mathcal{S} \setminus \{k_1\}$ and all the $(\alpha - 1)$ elements in $\mathcal{S} \setminus \{k_1\}$ are immediately after $k_1$ in $\hat{\bm{u}}$ (and any of its circular shifts in $\mathcal{U}_{(\bm{d}, \bm{f})}$). This is because we are considering circular demands where user $k_1$ requests for the file class $\mathcal{D}_{k_1} = \mathcal{S}$. Hence, when we consider all possible values of $\ell$, the subfile $W_{i, \mathcal{S}, \mathcal{T}}$ is counted a total of $\sum_{\ell = t}^{\alpha - 1} a_\ell (K - \ell)$ times in~\eqref{eqn: Complete Lower Bound} when we focus on circular demands with $f_{k_1} = i$ and $\mathcal{D}_{k_1} = \mathcal{S}$. The term
\begin{equation}
    a_\ell \coloneqq t!(\alpha - 1 - t)!(K - \alpha)!\binom{\ell - 1}{t - 1}f^{K - 1}
\end{equation}
counts the total number of vectors $\hat{\bm{u}}$ (and consequently of circular demands) for which $f_{k_1} = i$, $\mathcal{D}_{k_1} = \mathcal{S}$ and $\max_{j \in \mathcal{T}} \left| \bigcup_{i = \hat{u}(k_1) + 1}^{\hat{u}(j)} \{\hat{u}_i\} \right| = \ell$. Since the same reasoning applies whenever the file $W_{i, \mathcal{S}, \mathcal{T}}$ is requested by any of the other $(\alpha - t - 1)$ users in $\mathcal{S} \setminus \{\mathcal{T}, k_1\}$, namely, when $f_k = i$ and $\mathcal{D}_k$ for every $k \in \mathcal{S} \setminus \{\mathcal{T}, k_1\}$, we can conclude that the subfile $W_{i, \mathcal{S}, \mathcal{T}}$ appears a total of $(\alpha - t)\sum_{\ell = t}^{\alpha - 1}a_\ell (K - \ell)$ times in \eqref{eqn: Complete Lower Bound} when we consider all circular demands in $\mathcal{C}$. Moreover, the same reasoning applies to any other subfile. Hence, the expression in \eqref{eqn: Complete Lower Bound} simplifies as
\begin{align}
    R^\star & \geq \frac{1}{B\sum_{(\bm{d}, \bm{f}) \in \mathcal{C}} |\mathcal{U}_{(\bm{d}, \bm{f})}|} \sum_{(\bm{d}, \bm{f}) \in \mathcal{C}}  \sum_{\bm{u} \in \mathcal{U}_{(\bm{d}, \bm{f})}}\sum_{k \in [K]} \sum_{\mathcal{T} \subseteq [K] \setminus \{u_{1}, \dots, u_{k}\} \cap \mathcal{D}_{u_{k}}} \left| W_{f_{u_{k}}, \mathcal{D}_{u_{k}}, \mathcal{T}} \right| \\
            & = \frac{1}{Bf^K K!} \sum_{(\bm{d}, \bm{f}) \in \mathcal{C}}  \sum_{\bm{u} \in \mathcal{U}_{(\bm{d}, \bm{f})}}\sum_{k \in [K]} \sum_{\mathcal{T} \subseteq [K] \setminus \{u_{1}, \dots, u_{k}\} \cap \mathcal{D}_{u_{k}}} \left| W_{f_{u_{k}}, \mathcal{D}_{u_{k}}, \mathcal{T}} \right| \\
            & = \sum_{t = 0}^{\alpha} f(t) x_{t}
\end{align}
where we defined 
\begin{align}
    c_t & \coloneqq \frac{(\alpha - t)}{f^K K!}\sum_{\ell = t}^{\alpha - 1} a_\ell (K - \ell) \\
    f(t) & \coloneqq N c_t \\
    0 \leq x_{t} & \coloneqq \sum_{\mathcal{S} \subseteq [K] : |\mathcal{S}| = \alpha} \sum_{\mathcal{T} \subseteq \mathcal{S} : |\mathcal{T}| = t} \sum_{i \in [f]} \frac{\left| W_{i, \mathcal{S}, \mathcal{T}} \right|}{NB}.
\end{align}

At this point, we seek to lower bound the minimum worst-case load $R^{\star}(t)$ by solving the following optimization problem
\begin{subequations}\label{eqn: Optimization Problem}
  \begin{alignat}{2}
    &\min_{\bm{x}}  & \quad &\sum_{t = 0}^{\alpha} f(t) x_{t}\\
    &\text{subject to}  &  & \sum_{t = 0}^{\alpha}x_{t} = 1 \label{eqn: File Size Constraint}\\
    & & & \sum_{t = 0}^{\alpha}tx_{t} \leq \frac{KM}{N} \label{eqn: Memory Size Constraint}
  \end{alignat}
\end{subequations}
where \eqref{eqn: File Size Constraint} and \eqref{eqn: Memory Size Constraint} correspond to the file size constraint and the cumulative cache size constraint, respectively.

\subsubsection{Simplifying and Solving the Optimization Problem}
Before proceeding to solve the optimization problem, we wish to further simplify the coefficients $c_t$ and thus also $f(t)$. Indeed, this coefficient $c_{t}$ can be rewritten as
\begin{equation}
  c_{t} = \frac{1}{N\binom{\alpha}{t}}\sum_{\ell = t}^{\alpha - 1}\binom{\ell - 1}{t - 1}(K - \ell)
\end{equation}
recalling that $N = f\binom{K}{\alpha}$. Then, we can write
\begin{align}
  c_{t} & = \frac{1}{N\binom{\alpha}{t}}\sum_{\ell = t}^{\alpha - 1}\binom{\ell - 1}{t - 1}(K - \ell)\\
        & = \frac{1}{N\binom{\alpha}{t}}\left(K\sum_{\ell = t}^{\alpha - 1}\binom{\ell - 1}{t - 1} - \sum_{\ell = t}^{\alpha - 1}\ell\binom{\ell - 1}{t - 1}\right)\\
        & = \frac{1}{N\binom{\alpha}{t}}\left(K\sum_{\ell = t - 1}^{\alpha - 2}\binom{\ell}{t - 1} - t\sum_{\ell = t}^{\alpha - 1}\binom{\ell}{t}\right) \\
        & = \frac{K\binom{\alpha - 1}{t} - t\binom{\alpha}{t + 1}}{N\binom{\alpha}{t}}\label{eqn: Application of Hockey-Stick Identity} \\
        & = \frac{\binom{\alpha}{t + 1} + (K - \alpha)\binom{\alpha - 1}{t}}{N\binom{\alpha}{t}}
\end{align}
where \eqref{eqn: Application of Hockey-Stick Identity} uses the well-known hockey-stick identity which states that
\begin{equation}
    \sum_{i = k}^{n}\binom{i}{k} = \binom{n + 1}{k + 1}, \quad \forall n, k \in \mathbb{N}, \quad n \geq k.
\end{equation}
At this point, we can rewrite $f(t)$ as
\begin{equation}
    f(t) = \frac{\binom{\alpha}{t + 1} + (K - \alpha)\binom{\alpha - 1}{t}}{\binom{\alpha}{t}}.
\end{equation}

Now, since the auxiliary variable $x_{t}$ can be considered as a probability mass function, the optimization problem in \eqref{eqn: Optimization Problem} can be seen as the minimization of $\mathbb{E}[f(t)]$. Moreover, the following holds.
\begin{lemma}\label{lem: Strictly Decreasing Sequence}
  The function $f(t)$ is convex and is strictly decreasing for increasing values of $t$.
\end{lemma}
\begin{proof}
  The proof is reported in \refappendix{app: Proof of the Strictly Decreasing Sequence Lemma}.
\end{proof}
Taking advantage of \Cref{lem: Strictly Decreasing Sequence}, we can write $\mathbb{E}[f(t)] \geq f(\mathbb{E}[t])$ using Jensen's inequality. Then, since $f(t)$ is strictly decreasing with increasing $t \in [0 : \alpha]$, we can further write $f(\mathbb{E}[t]) \geq f\left(\frac{KM}{N}\right)$ taking advantage of the fact that $\mathbb{E}[t]$ is upper bounded as in \eqref{eqn: Memory Size Constraint}. Consequently, $\mathbb{E}[f(t)] \geq f\left(\frac{KM}{N}\right)$, and thus for $t = KM/N$ the converse bound is a piece-wise linear curve with corner points
\begin{equation}
  (M, R_{\text{LB}}) = \left(t\frac{N}{K}, \frac{\binom{\alpha}{t + 1} + (K - \alpha)\binom{\alpha - 1}{t}}{\binom{\alpha}{t}}\right), \quad \forall t \in [0 : \alpha].
\end{equation}
Thus, $R_{\text{LB}}$ takes the form
\begin{align}
    R_{\text{LB}} & = \frac{\binom{\alpha}{t + 1} + (K - \alpha)\binom{\alpha - 1}{t}}{\binom{\alpha}{t}} \\
                  & = \frac{\alpha - t}{1 + t} + (K - \alpha)\left(1 - \frac{t}{\alpha}\right) \\
                  & = \frac{\alpha(1 - \gamma_\alpha)}{1 + K\gamma} + K(1 - \gamma_\alpha) - \alpha(1 - \gamma_\alpha) \\
                  & = K(1 - \gamma_\alpha)\left[ 1 + \frac{\alpha}{K(1 + K\gamma)} - \frac{\alpha}{K} \right] \\
                  & = \frac{K(1 - \gamma_\alpha)}{1 + K\gamma} \Big[1 + K\gamma + \frac{\alpha}{K} - \frac{\alpha}{K}(1 + K\gamma)\Big] \\
                  & = \frac{K(1 - \gamma_\alpha)}{1 + K\gamma} \Big[(K - \alpha)\gamma + 1 \Big]
\end{align}
which completes the proof.

\subsection{A Detailed Example for the Converse Bound}\label{sec: An Exhaustive Example for the Converse Bound}

We present here in detail an example that can help the reader better understand the construction of the converse bound.

Let us consider the $(6, 4, 1)$ FDS structure which involves a file library $\mathcal{L} = \{W_{1, \mathcal{S}} : \mathcal{S} \subseteq [6], |\mathcal{S}| = 4 \}$ consisting of $C = \binom{K}{\alpha} = \binom{6}{4} = 15$ classes of files. Since there is only $f = 1$ file per class, there is a total of $N = fC = 15$ files, hence in this example we make no distinction between files and classes of files. Thus, for simplicity, we will here refer to file $W_{1, \mathcal{S}}$ directly as $W_{\mathcal{S}}$, which means that each file is entirely described by a $4$-tuple $\mathcal{S} \subseteq [6]$, and each demand instance is entirely defined by the $\bm{d} = (\mathcal{D}_1, \dots, \mathcal{D}_K)$ vector only. The FDS of each user $k \in [6]$ is given by
\begin{equation}
  \mathcal{F}_{k} = \{W_{S} : \mathcal{S} \subseteq [6], |\mathcal{S}| = 4, k \in \mathcal{S}\}
\end{equation}
and it consists of $f\binom{K - 1}{\alpha - 1} = 10$ files. Hence, this example considers $M \in [0 : 10]$, simply because having $M \geq 10$ implies that each user can preemptively cache the entirety of its FDS, which in turn implies $R^\star\left(|\mathcal{F}|\right) = 0$.

We start by assuming the most general uncoded and selfish cache placement where each file $W_{\mathcal{S}}$ is split into a total of  $2^{|\mathcal{S}|} = 2^{\alpha} = 16$ disjoint subfiles as
\begin{equation}
    W_{\mathcal{S}} = \{W_{\mathcal{S}, \mathcal{T}} : \mathcal{T} \subseteq \mathcal{S}\}, \quad \forall \mathcal{S} \subseteq [6] : |\mathcal{S}| = 4
\end{equation}
where we recall that $W_{\mathcal{S}, \mathcal{T}}$ is the subfile of $W_\mathcal{S}$ cached by the users in $\mathcal{T}$, and where we recall that the selfishness condition is guaranteed by forcing $\mathcal{T} \subseteq \mathcal{S}$.

\subsubsection{Constructing the Index Coding Problem}

For any given demand $\bm{d} = (\mathcal{D}_1, \dots, \mathcal{D}_K)$ where user $k$ asks for $W_{\mathcal{D}_{k}}$, we consider the index coding problem with $K' = K = 6$ users and $N' = K2^{\alpha - 1} = 48$ messages, where each user $k \in [6]$ has a desired message set
\begin{equation}
  \mathcal{M}_{k} = \{W_{\mathcal{D}_{k}, \mathcal{T}} : \mathcal{T} \subseteq \mathcal{D}_{k}, k \notin \mathcal{T} \}
\end{equation}
and a side information set
\begin{equation}
  \mathcal{A}_{k} = \{W_{\mathcal{S}, \mathcal{T}} : \mathcal{S} \subseteq [K], |\mathcal{S}| = \alpha, \mathcal{T} \subseteq \mathcal{S}, k \in \mathcal{S} \cap \mathcal{T}\}.
\end{equation}
The side information graph of this index coding problem is the directed graph where each desired subfile represents a graph vertex and where a connection exists from vertex $W_{\mathcal{D}_{k_{1}}, \mathcal{T}_{1}}$ to $W_{\mathcal{D}_{k_{2}}, \mathcal{T}_{2}}$ if and only if $W_{\mathcal{D}_{k_{1}}, \mathcal{T}_{1}} \in \mathcal{A}_{k_{2}}$. To now create an acyclic subgraph of the above graph, we recall from~\cite[Lemma 1]{8963629} that the set
\begin{equation}\label{eqn: Acyclic Set of Vertices 2}
  \bigcup_{k \in [K]} \bigcup_{\mathcal{T} \subseteq [K] \setminus \{u_{1}, \dots, u_{k}\} \cap \mathcal{D}_{u_{k}}} \left\{W_{\mathcal{D}_{u_{k}}, \mathcal{T}}\right\}
\end{equation}
does not contain any directed cycle for any demand $\bm{d}$ and user permutation $\bm{u}$. This, together with \Cref{thm: Acyclic Subgraph Converse Bound}, implies that
\begin{equation}\label{eqn: Index Coding Lower Bound 2}
  BR^{\star} \geq \sum_{k \in [K]} \sum_{\mathcal{T} \subseteq [K] \setminus \{u_{1}, \dots, u_{k}\} \cap \mathcal{D}_{u_{k}}} \left| W_{\mathcal{D}_{u_{k}}, \mathcal{T}} \right|.
\end{equation}

\subsubsection{Selection of Distinct Demands}

To render the above bound meaningful, we need to carefully select a set of demands that will eventually help us symmetrize the problem as well as render the bound tighter. Toward this, we create several lower bounds as the one in \eqref{eqn: Index Coding Lower Bound 2}, one for each chosen demand and user permutation. Our desired symmetry is achieved by considering only the set of circular demands $\mathcal{C}$ and, for each $\bm{d} \in \mathcal{C}$, the user permutations in $\mathcal{U}_{\bm{d}}$, where this last set simply contains the $K$ circular shifts of the vector $\hat{\bm{u}}$ associated to each circular demand.
For example, in our $(6, 4, 1)$ FDS scenario, one such circular demand is $\bm{d} = (1234, 2345, 3456, 1456, 1256, 1236)$, because it satisfies the condition $\cup_{i = k + 1}^{k + \alpha - 1} \{\hat{u}_{i}\} = \mathcal{D}_{\hat{u}_{k}} \setminus \{\hat{u}_{k}\}$ for each $k \in [6]$ with the vector $\hat{\bm{u}} = (1, 2, 3, 4, 5, 6)$. This same condition is also satisfied by all $K$ circular shifts of $\hat{\bm{u}} = (1, 2, 3, 4, 5, 6)$. For this demand $\bm{d} = (1234, 2345, 3456, 1456, 1256, 1236)$, we will thus construct $6$ bounds as in~\eqref{eqn: Index Coding Lower Bound 2}.

By averaging all such bounds, a new lower bound on $R^\star$ appears in the following form
\begin{equation}\label{eqn: Complete Lower Bound 2}
    BR^\star \sum_{\bm{d} \in \mathcal{C}} |\mathcal{U}_{\bm{d}}| \geq \sum_{\bm{d} \in \mathcal{C}}  \sum_{\bm{u} \in \mathcal{U}_{\bm{d}}} \sum_{k \in [K]} \sum_{\mathcal{T} \subseteq [K] \setminus \{u_{1}, \dots, u_{k}\} \cap \mathcal{D}_{u_{k}}} \left| W_{\mathcal{D}_{u_{k}}, \mathcal{T}} \right|.
\end{equation}

\subsubsection{Counting the Selected Demands}

To simplify the above, we proceed to count the total number of circular demands. Let us focus without loss of generality on user~$1$ and also on those circular demands where user~$1$ requests the file $W_{\mathcal{S}}$ such that $\mathcal{S} = \mathcal{D}_1 = \{1, 2, 3, 4\}$. We can see that there exists a total of $(\alpha - 1)!(K - \alpha)! = 12$ circular demands with this $\mathcal{D}_1$. If we consider $\hat{u}_{1} = 1$ without loss of generality, such demands are all those associated to a vector $\hat{\bm{u}}$ where $\cup_{i = 2}^{\alpha}\{\hat{u}_{i}\} = \mathcal{D}_{1} \setminus \{1\} = \{2, 3, 4\}$ and $\cup_{i = \alpha + 1}^{K}\{\hat{u}_{i}\} = [K] \setminus \mathcal{D}_{1} = \{5, 6\}$. Given that user~$1$ can ask for a file among a total of $|\mathcal{F}_1| = \binom{K - 1}{\alpha - 1} = 10$ files, we can conclude that the total number of such circular demands is equal to $|\mathcal{C}| = |\mathcal{F}_1|(\alpha - 1)!(K - 1)! = 120$. For each such demand we can build $K$ index coding bounds, one for each of the $K$ circular shifts corresponding to the vector $\hat{\bm{u}}$, thus resulting in a total of $\sum_{\bm{d} \in \mathcal{C}} |\mathcal{U}_{\bm{d}}| = 720$ lower bounds --- as the one in \eqref{eqn: Index Coding Lower Bound 2} --- being used for the expression in \eqref{eqn: Complete Lower Bound 2}.

\subsubsection{Constructing the Optimization Problem}

Let us keep our focus on user~$1$ and again on those circular demands where user~$1$ requests the file $W_{\mathcal{S}}$ with $\mathcal{S} = \{1, 2, 3, 4\} = \mathcal{D}_1$. Since there is a one-to-one correspondence between the $(\alpha - 1)!(K - \alpha)! = 12$ circular demands and the set of $K$ circular shifts of the ordered vector of users $\hat{\bm{u}}$, then --- assuming without loss of generality that $\hat{u}_{1} = 1$ --- the vectors $\hat{\bm{u}}$ for all the $12$ demands take the form
\begin{alignat}{2}
  \hat{\bm{u}}_{1} & = (1, 2, 3, 4, 5, 6) & \qquad \hat{\bm{u}}_{7} & = (1, 2, 3, 4, 6, 5) \\
  \hat{\bm{u}}_{2} & = (1, 3, 2, 4, 5, 6) & \qquad \hat{\bm{u}}_{8} & = (1, 3, 2, 4, 6, 5) \\
  \hat{\bm{u}}_{3} & = (1, 3, 4, 2, 5, 6) & \qquad \hat{\bm{u}}_{9} & = (1, 3, 4, 2, 6, 5) \\
  \hat{\bm{u}}_{4} & = (1, 4, 3, 2, 5, 6) & \qquad \hat{\bm{u}}_{10} & = (1, 4, 3, 2, 6, 5) \\
  \hat{\bm{u}}_{5} & = (1, 4, 2, 3, 5, 6) & \qquad \hat{\bm{u}}_{11} & = (1, 4, 2, 3, 6, 5) \\
  \hat{\bm{u}}_{6} & = (1, 2, 4, 3, 5, 6) & \qquad \hat{\bm{u}}_{12} & = (1, 2, 4, 3, 6, 5).
\end{alignat}
The above one-to-one correspondence means that each vector $\hat{\bm{u}}$ above corresponds to a circular demand with $\mathcal{D}_1 = \{1, 2, 3, 4\}$. For example, vector $\hat{\bm{u}}_1$ corresponds to the demand $\bm{d}_1 = (1234, 2345, 3456, 1456, 1256, 1236)$. In addition, for each vector $\hat{\bm{u}}$, we consider as user permutations the $K$ circular shifts of $\hat{\bm{u}}$. For instance, when we consider the circular demand associated to the vector $\hat{\bm{u}}_1$, the user permutations are given by the following set
\begin{equation}
    \begin{aligned}
        \mathcal{U}_{\bm{d}_1} = \{ &(1, 2, 3, 4, 5, 6), (2, 3, 4, 5, 6, 1), (3, 4, 5, 6, 1, 2), \\
        & (4, 5, 6, 1, 2, 3), (5, 6, 1, 2, 3, 4), (6, 1, 2, 3, 4, 5)\}.
    \end{aligned}
\end{equation}
Consider the subfile $W_{\mathcal{S}, \mathcal{T}}$ with $\mathcal{T} \subseteq \mathcal{S} \setminus \{1\}$ and $|\mathcal{T}| = t$. Following the same line of reasoning as in~\Cref{sec: Proof of the Lower Bound}, and taking advantage of \Cref{lem: Circular Shift Lemma} as well as focusing on circular demands with $\mathcal{D}_1 = \mathcal{S} = \{1, 2, 3, 4\}$, we see that this subfile is counted $(K - \ell)$ times when we consider the $K$ circular shifts of each $\hat{\bm{u}}$, where
\begin{equation}
    \ell = \max_{j \in \mathcal{T}} \left| \bigcup_{i = 2}^{\hat{u}(j)} \{\hat{u}_i\} \right|
\end{equation}
and where we again used that $\hat{u}(1) = 1$. For example, consider $\hat{\bm{u}}_{4} = (1, 4, 3, 2, 5, 6)$ and the subfile $W_{\mathcal{S}, 24}$. Since $\ell = \max_{j \in \mathcal{T}} \left| \bigcup_{i = 2}^{\hat{u}(j)} \{\hat{u}_i\} \right| = \hat{u}(2) - \hat{u}(1) = 3$, the subfile $W_{\mathcal{S}, 24}$ is counted a total of $(K - \ell) = 3$ times across the $K$ circular shifts of $\hat{\bm{u}}_{4}$.

Considering that $\ell \in [t : \alpha - 1] = [t : 3]$, as already explained in the general description of the main proof of the converse, we see that each subfile $W_{\mathcal{S}, \mathcal{T}}$ with $|\mathcal{T}| = t$ and $\mathcal{T} \subseteq \mathcal{S} \setminus \{1\}$ is counted a total of $\sum_{\ell = t}^{3}a_{\ell}(K - \ell)$ times in~\eqref{eqn: Complete Lower Bound 2} when going over the circular demands having $\mathcal{D}_1 = \mathcal{S} = \{1, 2, 3, 4\}$. The term
\begin{equation}
  a_{\ell} = t!(\alpha - 1 - t)!(K - \alpha)!\binom{\ell - 1}{t - 1} = t!(3 - t)!2!\binom{\ell - 1}{t - 1}
\end{equation}
counts the total number of vectors $\hat{\bm{u}}$ (and consequently the number of circular demands $\bm{d} \in \mathcal{C}$) corresponding to $\mathcal{D}_{1} = \{1, 2, 3, 4\}$ and $\max_{j \in \mathcal{T}} \left| \bigcup_{i = 2}^{\hat{u}(j)} \{\hat{u}_i\} \right| = \ell$. Since the same reasoning applies whenever the file $W_{\mathcal{S}}$ is requested by any other of the $(\alpha - 1 - t)$ users in $\mathcal{S} \setminus \{\mathcal{T}, 1\}$, we can see that each subfile $W_{\mathcal{S}, \mathcal{T}}$ appears $(4 - t)\sum_{\ell = t}^{3}a_{\ell}(K - \ell)$ times in \eqref{eqn: Complete Lower Bound 2} when we consider all circular demands in $\mathcal{C}$. Similarly, the same reasoning applies to any other file in $\mathcal{L}$. Consequently, the expression in \eqref{eqn: Complete Lower Bound 2} simplifies as
\begin{align}
    R^\star & \geq \frac{1}{B\sum_{\bm{d} \in \mathcal{C}} |\mathcal{U}_{\bm{d}}|} \sum_{\bm{d} \in \mathcal{C}}  \sum_{\bm{u} \in \mathcal{U}_{\bm{d}}} \sum_{k \in [K]} \sum_{\mathcal{T} \subseteq [K] \setminus \{u_{1}, \dots, u_{k}\} \cap \mathcal{D}_{u_{k}}} \left| W_{\mathcal{D}_{u_{k}}, \mathcal{T}} \right| \\
    & = \frac{1}{720B} \sum_{\bm{d} \in \mathcal{C}} \sum_{\bm{u} \in \mathcal{U}_{\bm{d}}} \sum_{k \in [K]} \sum_{\mathcal{T} \subseteq [K] \setminus \{u_{1}, \dots, u_{k}\} \cap \mathcal{D}_{u_{k}}} \left| W_{\mathcal{D}_{u_{k}}, \mathcal{T}} \right| \\
    & = \sum_{t = 0}^{\alpha}f(t)x_{t}
\end{align}
where we defined
\begin{align}
  c_{t} & \coloneqq \frac{(4 - t)}{720}\sum_{\ell = t}^{\alpha - 1}a_{\ell}(K - \ell) \\
  f(t) & \coloneqq Nc_{t} \\
  0 \leq x_{t} & \coloneqq \sum_{\mathcal{S} \subseteq [K] : |\mathcal{S}| = \alpha} \sum_{\mathcal{T} \subseteq \mathcal{S} : |\mathcal{T}| = t} \frac{\left| W_{\mathcal{S}, \mathcal{T}} \right|}{NB}.
\end{align}
At this point we can formulate the optimization problem as in \eqref{eqn: Optimization Problem} and solve it to obtain the converse.

\subsubsection{Solving the Optimization Problem}

Since the variable $x_{t}$ can be interpreted as a probability mass function and since --- as we recall from~\Cref{lem: Strictly Decreasing Sequence} --- the coefficients $f(t)$ represent a strictly decreasing convex sequence, we can conclude that the optimization problem can be easily solved by using Jensen's inequality and the cumulative cache size constraint. As shown in \Cref{sec: Proof of the Lower Bound}, the coefficients can be rewritten in the following form
\begin{equation}
  f(t) =  \frac{\binom{\alpha}{t + 1} + (K - \alpha)\binom{\alpha - 1}{t}}{\binom{\alpha}{t}}
\end{equation}
so we can obtain that the converse bound is a piece-wise linear curve with corner points $\left(t\frac{5}{2}, \frac{\binom{4}{t + 1} + 2\binom{3}{t}}{\binom{4}{t}}\right)$ for every $t \in [0 : 4]$.

\section{The Exact Memory-Load Trade-Off for~\texorpdfstring{$\alpha$}{alpha}-Demands}\label{sec: The Scheme for alpha-Demands}

We will here draw insights from the converse to establish a general cache placement policy, and then a delivery scheme that applies to a specific set of so-called \texorpdfstring{$\alpha$}{alpha}-Demands. For these demands and for the specific placement policy, the scheme will be proven optimal with the use of an additional converse. 

We start by noticing that the converse in~\Cref{thm: Lower Bound for the FDS Structure} can be decomposed as 
\begin{align} \label{eq:insight1}
  R_{\text{LB}}(t) & = \frac{\binom{\alpha}{t + 1}}{\binom{\alpha}{t}} + \frac{(K - \alpha)\binom{\alpha - 1}{t}}{\binom{\alpha}{t}}
\end{align}
with the second term $R_{2}(t) \coloneqq \frac{(K - \alpha)\binom{\alpha - 1}{t}}{\binom{\alpha}{t}} = (K-\alpha)(1-\gamma_\alpha)$ bringing to mind uncoded delivery to $(K - \alpha)$ users, and with the first term $R_{1}(t) \coloneqq \frac{\binom{\alpha}{t + 1}}{\binom{\alpha}{t}} = \frac{\alpha(1 - \gamma_\alpha)}{1 + \alpha\gamma_\alpha}$ bringing to mind a smaller MAN placement-and-delivery (unselfish) problem with $\alpha$ users, a common library, and normalized cache size $\gamma_\alpha$. Let us exploit this observation to suggest a placement. 

\subsection{Cache Placement}\label{sec: Cache Placement Procedure}

As noted, we can think of the $ R_{1}(t)$ term as representing the optimal load in a ``smaller'' $\alpha$-MAN problem with $\alpha$ users that \emph{are known in advance} to be interested in a common class of files and thus benefit from the corresponding $\alpha$-user MAN placement. If each user --- as is the case in our setting --- can allocate a fraction $\gamma_\alpha$ for each file of potential interest, then a MAN placement implies that each user stores a total of $f\gamma_\alpha$ files\footnote{Recall that $f$ is the total number of files in this common class of files.}. Here, in our effort to provide a placement method, we must account for the fact that there is a total of $\binom{K}{\alpha}$ such ``smaller'' MAN problems, because there are $C = \binom{K}{\alpha}$ file classes. Let us now recall that each user appears in a total of $\binom{K-1}{\alpha-1}$ such smaller problems, since there are $\binom{K-1}{\alpha-1}$ file classes that each user is interested in. Our placement must account for the possibility of each user participating in any such smaller problem. This requires each user to store $f\gamma_\alpha$ files per class of interest, and thus requires a total storage capacity of $\gamma_\alpha f\binom{K - 1}{\alpha - 1} = M$, which, as we see, nicely satisfies the cache size constraint. This reasoning justifies the cache placement procedure that we present below.

In our proposed uncoded and selfish cache placement method, based on the same combinatorial argument of the MAN scheme, each user $k$ proceeds to cache only from $\mathcal{F}_{k}$. The placement begins by splitting each file into $\binom{\alpha}{t}$ non-overlapping subfiles as
\begin{equation}
  W_{i, \mathcal{S}} = \{W_{i, \mathcal{S}, \mathcal{T}} : \mathcal{T} \subseteq \mathcal{S}, |\mathcal{T}| = t\}, \quad \forall \mathcal{S} \subseteq [K] : |\mathcal{S}| = \alpha, \quad \forall i \in [f]
\end{equation}
and then is completed by filling the cache $\mathcal{Z}_{k}$ of each user $k \in [K]$ as
\begin{equation}
  \mathcal{Z}_{k} = \{W_{i, \mathcal{S}, \mathcal{T}}: i \in [f], \mathcal{S} \subseteq [K], |\mathcal{S}| = \alpha, \mathcal{T} \subseteq \mathcal{S}, |\mathcal{T}| = t, k \in \mathcal{S} \cap \mathcal{T}\}.
\end{equation}
Each cache stores $\binom{\alpha - 1}{t - 1}$ subfiles for each file in its FDS, thus abiding by the cache size constraint 
\begin{equation}
  |\mathcal{F}|\binom{\alpha - 1}{t - 1}\frac{B}{\binom{\alpha}{t}} = f\binom{K - 1}{\alpha - 1} \frac{t}{\alpha} B = MB.
\end{equation}

Next we describe the delivery scheme for a specific set of demands. Unfortunately, the above reasoning does not immediately reflect --- at least not to us --- a universal delivery solution for any set of demands. To the best of our understanding, our cache placement introduces the need to resolve a large number of non-isomorphic index coding problems. What the above observation does allow though is insight for the delivery for a specific class of demands, as we see below. 

\subsection{Delivery Scheme for the Set of \texorpdfstring{$\alpha$}{alpha}-Demands}

We now present the delivery method for the following class of demands.

\begin{definition}[$\alpha$-Demands]
  Considering the $(K, \alpha, f)$ FDS structure with $f \geq \alpha$, the demand defined by the vectors $\bm{d} = (\mathcal{D}_{1}, \dots, \mathcal{D}_{K})$ and $\bm{f} = (f_{1}, \dots, f_{K})$ is an $\alpha$-demand if and only if there exists at least one set of users $\mathcal{K} \subseteq [K]$ such that $|\mathcal{K}| = \alpha$, $f_{k_1} \neq f_{k_2}$ for any $k_1 \neq k_2$ with $k_1, k_2 \in \mathcal{K}$ and $\mathcal{D}_{k} = \mathcal{K}$ for all $k \in \mathcal{K}$.
\end{definition}

Such demands can exist only if $f \geq \alpha$.  Indeed  if we have at least $\alpha$ files per class, then we can have distinct demands where there exists at least one set of $\alpha$ users requesting distinct files, all belonging to the same file class.

Let $R_{\alpha, \text{c}}$ denote the worst-case load when only $\alpha$-demands are considered, and when the cache placement in~\Cref{sec: Cache Placement Procedure} is adopted. We are now ready to provide the exact characterization of optimal such load $R^\star_{\alpha, \text{c}}$.

\begin{proposition}[The Exact Memory-Load Trade-Off for $\alpha$-Demands Under the Presented Symmetric Placement]\label{thm: The Exact Worst-Case Load for alpha-Demands}
  Assuming the selfish and uncoded cache placement presented in \Cref{sec: Cache Placement Procedure}, the optimal worst-case communication load $R^{\star}_{\alpha, \text{c}}$ for the $(K, \alpha, f)$ FDS structure and $\alpha$-Demands is a piece-wise linear curve with corner points
  \begin{equation}
    (M, R^{\star}_{\alpha, \text{c}}) = \left(t \frac{N}{K}, \frac{\binom{\alpha}{t + 1} + (K - \alpha)\binom{\alpha - 1}{t}}{\binom{\alpha}{t}}\right), \quad \forall t \in [0 : \alpha]
  \end{equation}
  again corresponding to
  \begin{equation}
    R^{\star}_{\alpha, \text{c}} = \frac{K(1 - \gamma_\alpha)}{K\gamma+1}\Big[(K - \alpha)\gamma+1\Big].
  \end{equation}
\end{proposition}

\begin{proof}
  The proof of the converse is reported in~\refappendix{app: Converse Proof of the Exact Worst-Case Load for alpha-Demands}, whereas the proof of the achievability is reported below in \Cref{sec: Achievability Proof of the Exact Worst-Case Load for alpha-Demands}.
\end{proof} 

\subsection{Achievability Proof of \texorpdfstring{\Cref{thm: The Exact Worst-Case Load for alpha-Demands}}{Proposition \ref{thm: The Exact Worst-Case Load for alpha-Demands}}}\label{sec: Achievability Proof of the Exact Worst-Case Load for alpha-Demands}
By definition, any $\alpha$-demand has at least one set of $\alpha$ users requesting distinct files from the same file class. If we denote by $\mathcal{K}$ one of such sets, it holds that $\mathcal{D}_{k} = \mathcal{K}$ for all $k \in \mathcal{K}$ and $f_{k_1} \neq f_{k_2}$ for all $k_1 \neq k_2$ with $k_1, k_2 \in \mathcal{K}$. Consider user $k \in \mathcal{K}$. According to the cache placement procedure in \Cref{sec: Cache Placement Procedure}, this user does not have in its cache any subfile $W_{f_{k}, \mathcal{K}, \mathcal{T}}$ where $\mathcal{T} \subseteq \mathcal{K}$, $|\mathcal{T}| = t$ and $k \notin \mathcal{T}$. If we focus on this set $\mathcal{K}$ of $\alpha$ users only, we can automatically construct the following sequence of multicast messages
\begin{equation}
    X_{\mathcal{K}} = \left(\bigoplus_{k \in \mathcal{S}}W_{f_{k}, \mathcal{K}, \mathcal{S} \setminus \{k\}} : \mathcal{S} \subseteq \mathcal{K}, |\mathcal{S}| = t + 1\right).
\end{equation}
For the remaining $(K - \alpha)$ users in $[K] \setminus \mathcal{K}$, we consider the following sequence 
\begin{equation}
    X_{[K] \setminus \mathcal{K}} = \left(W_{f_k, \mathcal{D}_k, \mathcal{T}} : k \in [K] \setminus \mathcal{K}, k \notin \mathcal{T} \right)
\end{equation}
of uncoded transmissions. Then, the transmitter delivers the concatenated $X = (X_{\mathcal{K}}, X_{[K] \setminus \mathcal{K}})$, inducing a load 
\begin{equation}
    \frac{|X|}{B} = \frac{|X_{\mathcal{K}}| + |X_{[K] \setminus \mathcal{K}}| }{B} = \frac{\binom{\alpha}{t + 1} + (K - \alpha)\binom{\alpha - 1}{t}}{\binom{\alpha}{t}}
\end{equation}
which implies that $R^{\star}_{\alpha, \text{c}}(t) \leq \frac{\binom{\alpha}{t + 1} + (K - \alpha)\binom{\alpha - 1}{t}}{\binom{\alpha}{t}}$ for all $t \in [0 : \alpha]$.

\subsection{Example of the Achievable Scheme}

Consider the $(K, \alpha, f) = (5, 3, 3)$ FDS structure. We have $C = \binom{K}{\alpha} = \binom{5}{3} = 10$ classes of files with a total of $N = fC = 30$ files. The FDS structure is given by
\begin{align}
  \mathcal{F}_{1} & = \{\mathcal{W}_{123}, \mathcal{W}_{124}, \mathcal{W}_{125}, \mathcal{W}_{134}, \mathcal{W}_{135}, \mathcal{W}_{145}\} \\
  \mathcal{F}_{2} & = \{\mathcal{W}_{123}, \mathcal{W}_{124}, \mathcal{W}_{125}, \mathcal{W}_{234}, \mathcal{W}_{235}, \mathcal{W}_{245}\} \\
  \mathcal{F}_{3} & = \{\mathcal{W}_{123}, \mathcal{W}_{134}, \mathcal{W}_{135}, \mathcal{W}_{234}, \mathcal{W}_{235}, \mathcal{W}_{345}\} \\
  \mathcal{F}_{4} & = \{\mathcal{W}_{124}, \mathcal{W}_{134}, \mathcal{W}_{145}, \mathcal{W}_{234}, \mathcal{W}_{245}, \mathcal{W}_{345}\} \\
  \mathcal{F}_{5} & = \{\mathcal{W}_{125}, \mathcal{W}_{135}, \mathcal{W}_{145}, \mathcal{W}_{235}, \mathcal{W}_{245}, \mathcal{W}_{345}\}
\end{align}
where $\mathcal{W}_{\mathcal{S}} = \{W_{1, \mathcal{S}}, W_{2, \mathcal{S}}, W_{3, \mathcal{S}}\}$ for each triplet $\mathcal{S} \subseteq [5]$. Let us consider the scenario of $t = 2$. In this case, each file is split as
\begin{alignat}{2}
  W_{i, 123} & = \{W_{i, 123, 12}, W_{i, 123, 13}, W_{i, 123, 23}\} & \quad W_{i, 145} & = \{W_{i, 145, 14}, W_{i, 145, 15}, W_{i, 145, 45}\} \\
  W_{i, 124} & = \{W_{i, 124, 12}, W_{i, 124, 14}, W_{i, 124, 24}\} & \quad W_{i, 234} & = \{W_{i, 234, 23}, W_{i, 234, 24}, W_{i, 234, 34}\} \\
  W_{i, 125} & = \{W_{i, 125, 12}, W_{i, 125, 15}, W_{i, 125, 25}\} & \quad W_{i, 235} & = \{W_{i, 235, 23}, W_{i, 235, 25}, W_{i, 235, 35}\} \\
  W_{i, 134} & = \{W_{i, 134, 13}, W_{i, 134, 14}, W_{i, 134, 34}\} & \quad W_{i, 245} & = \{W_{i, 245, 24}, W_{i, 245, 25}, W_{i, 245, 45}\} \\
  W_{i, 135} & = \{W_{i, 135, 13}, W_{i, 135, 15}, W_{i, 135, 35}\} & \quad W_{i, 345} & = \{W_{i, 345, 34}, W_{i, 345, 35}, W_{i, 345, 45}\}
\end{alignat}
for all $i \in [3]$.

Consider the demand defined by $\bm{f} = (1, 2, 3, 1, 1)$ and $\bm{d} = (123, 123, 123, 124, 125)$. This is an $\alpha$-demand because there exists a set $\mathcal{K}$ of $\alpha$ users all requesting distinct files belonging to the same file class $\mathcal{K}$. Here this set is $\mathcal{K} = \{1, 2, 3\}$.

In accordance to the described selfish and uncoded cache placement, each user desires a total of $\binom{\alpha - 1}{t}$ subfiles which are not in its cache. In this case, each user simply desires $\binom{2}{2} = 1$ subfile. The delivery of these subfiles involves the following MAN XOR 
\begin{align}
  X_{123} & = \left(\bigoplus_{k \in \mathcal{S}}W_{f_{k}, \mathcal{K}, \mathcal{S} \setminus \{k\}} : \mathcal{S} \subseteq \mathcal{K}, |\mathcal{S}| = t + 1\right)\\
                    & = \left(W_{1, 123, 23} \oplus W_{2, 123, 13} \oplus W_{3, 123, 12}\right)
\end{align}
and then the following two uncoded transmissions
\begin{align}
  X_{45} & = (W_{1, 124, 12}, W_{1, 125, 12})
\end{align}
that serve the users outside $\mathcal{K}$. Given that the subpacketization is $\binom{\alpha}{t} = \binom{3}{2} = 3$, the transmitted signal $X = (X_{123}, X_{45})$ induces a communication load of $R_{\alpha, \text{c}}(2) = |X|/B = 1$ which matches the corresponding optimal $R^{\star}_{\alpha, \text{c}}(2)$ from~\Cref{thm: The Exact Worst-Case Load for alpha-Demands}.

\section{Additional Optimal Schemes for Circular Demands}\label{sec: Some Optimal Schemes for Circular Demands}

We here present schemes that optimally deliver circular demands. We will do so for the $(5, 4, f)$ FDS structure with $t \in \{2, 3\}$, and for the $(6, 5, f)$ FDS structure with $t = 3$. The optimal schemes assume the selfish and uncoded cache placement described in~\Cref{sec: Cache Placement Procedure}. We prove optimality simply by showing that the load provided by the proposed achievable schemes matches the converse bound in~\Cref{thm: Lower Bound for the FDS Structure}. This suffices because, as we might recall, the construction of the converse employed only circular demands\footnote{We urge the reader not to conclude from this statement that the main converse of this work holds only for circular demands. On the contrary, the converse bounds the worst-case load, without any consideration of the type of demand. The fact though that the construction of the converse employed solely circular demands allows us to use this same converse to prove that the communication load presented here is indeed the smallest among all possible delivery schemes for circular demands, even if we do not know whether circular demands belong to the class of worst-case demands or not.}.

\subsection{Circular Demands and the \texorpdfstring{$(5, 4, f)$}{(5, 4, f)} FDS Structure}

The scheme presented here is a generalization, for any circular demand, of the example in~\Cref{sec: Motivating Example}. For the considered $(K, \alpha, f) = (5, 4, f)$ structure, we know that there are $C = \binom{K}{\alpha} = 5$ file classes $\mathcal{W}_{1234}, \mathcal{W}_{1235}, \mathcal{W}_{1245}, \mathcal{W}_{1345}, \mathcal{W}_{2345}$, corresponding to $N = fC = 5f$ files. The $5$ FDSs take the form
\begin{align}
\mathcal{F}_{1} &= \{\mathcal{W}_{1234}, \mathcal{W}_{1235}, \mathcal{W}_{1245}, \mathcal{W}_{1345}\} \\
\mathcal{F}_{2} &= \{\mathcal{W}_{1234}, \mathcal{W}_{1235}, \mathcal{W}_{1245}, \mathcal{W}_{2345}\} \\
\mathcal{F}_{3} &= \{\mathcal{W}_{1234}, \mathcal{W}_{1235}, \mathcal{W}_{1345}, \mathcal{W}_{2345}\} \\
\mathcal{F}_{4} &= \{\mathcal{W}_{1234}, \mathcal{W}_{1245}, \mathcal{W}_{1345}, \mathcal{W}_{2345}\} \\
\mathcal{F}_{5} &= \{\mathcal{W}_{1235}, \mathcal{W}_{1245}, \mathcal{W}_{1345}, \mathcal{W}_{2345}\}
\end{align}
where we recall that $\mathcal{W}_{\mathcal{S}} = \{W_{i, \mathcal{S}} : i \in [f]\}$. The scheme works for any value of $f \in \mathbb{Z}^{+}$, and for any circular demand. Such general circular demand is identified by the vector $\hat{\bm{u}} = (\hat{u}_1, \hat{u}_2, \hat{u}_3, \hat{u}_4, \hat{u}_5)$ that must satisfy the property (cf.~\Cref{def: Circular Demands}) that $\cup_{i = k + 1}^{k + 3} \{\hat{u}_{i}\} = \mathcal{D}_{\hat{u}_{k}} \setminus \{\hat{u}_{k}\}$ for each $k \in [5]$. The FDS request graph of such generic circular demand is shown in~\Cref{fig: FDS Request Graph for a Generic Circular Demand and First FDS Structure}. Notice that, as expected, this is not a complete graph.

\begin{figure}[!htb]
\centering
\begin{tikzpicture}
  \node[regular polygon, regular polygon sides = 5, inner sep = 1.5cm](G){};
  \node[draw, circle, label = center:$W_{f_{\hat{u}_1}, \mathcal{D}_{\hat{u}_1}}$, minimum size = 1.6cm](1) at (G.corner 1){};
  \node[draw, circle, label = center:$W_{f_{\hat{u}_5}, \mathcal{D}_{\hat{u}_5}}$, minimum size = 1.6cm](5) at (G.corner 2){};
  \node[draw, circle, label = center:$W_{f_{\hat{u}_4}, \mathcal{D}_{\hat{u}_4}}$, minimum size = 1.6cm](4) at (G.corner 3){};
  \node[draw, circle, label = center:$W_{f_{\hat{u}_3}, \mathcal{D}_{\hat{u}_3}}$, minimum size = 1.6cm](3) at (G.corner 4){};
  \node[draw, circle, label = center:$W_{f_{\hat{u}_2}, \mathcal{D}_{\hat{u}_2}}$, minimum size = 1.6cm](2) at (G.corner 5){};
  \node at (1.north east)[font = \small, anchor = west]{$\hat{u}_1$};
  \node at (2.north east)[font = \small, anchor = west]{$\hat{u}_2$};
  \node at (3.south east)[font = \small, anchor = west]{$\hat{u}_3$};
  \node at (4.south west)[font = \small, anchor = east]{$\hat{u}_4$};
  \node at (5.north west)[font = \small, anchor = east]{$\hat{u}_5$};
  \draw[-stealth](1)--(2); \draw[-stealth](2)--(3); \draw[-stealth](3)--(4); \draw[-stealth](4)--(5); \draw[-stealth](5)--(1);
  \draw[stealth-stealth](1)--(4); \draw[stealth-stealth](2)--(5); \draw[stealth-stealth](3)--(1); \draw[stealth-stealth](4)--(2); \draw[stealth-stealth](5)--(3);
\end{tikzpicture}
\caption{FDS request graph for a generic circular demand identified by the vector $\hat{\bm{u}} = (\hat{u}_1, \hat{u}_2, \hat{u}_3, \hat{u}_4, \hat{u}_5)$ and the $(K, \alpha, f) = (5, 4, f)$ FDS structure.}
\label{fig: FDS Request Graph for a Generic Circular Demand and First FDS Structure}
\end{figure}
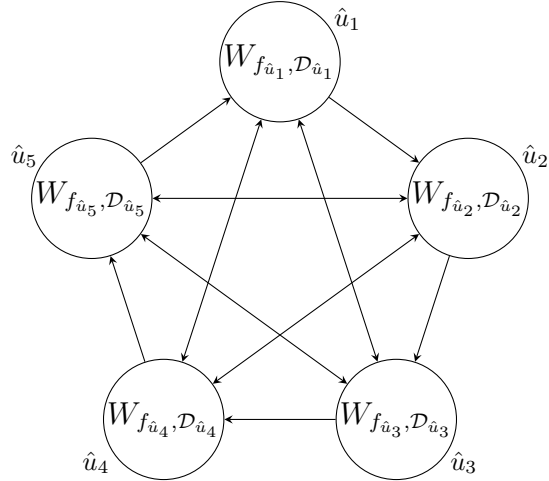

\subsubsection{The case of \texorpdfstring{$t = 2$}{t = 2}}

According to the cache placement in \Cref{sec: Cache Placement Procedure}, each file is split into $\binom{\alpha}{t} = \binom{4}{2} = 6$ non-overlapping subfiles as
\begin{align}
    W_{i, 1234} & = \{W_{i, 1234, 12}, W_{i, 1234, 13}, W_{i, 1234, 14}, W_{i, 1234, 23}, W_{i, 1234, 24}, W_{i, 1234, 34}\} \\
    W_{i, 1235} & = \{W_{i, 1235, 12}, W_{i, 1235, 13}, W_{i, 1235, 15}, W_{i, 1235, 23}, W_{i, 1235, 25}, W_{i, 1235, 35}\} \\
    W_{i, 1245} & = \{W_{i, 1245, 12}, W_{i, 1245, 14}, W_{i, 1245, 15}, W_{i, 1245, 24}, W_{i, 1245, 25}, W_{i, 1245, 45}\} \\
    W_{i, 1345} & = \{W_{i, 1345, 13}, W_{i, 1345, 14}, W_{i, 1345, 15}, W_{i, 1345, 34}, W_{i, 1345, 35}, W_{i, 1345, 45}\} \\
    W_{i, 2345} & = \{W_{i, 2345, 23}, W_{i, 2345, 24}, W_{i, 2345, 25}, W_{i, 2345, 34}, W_{i, 2345, 35}, W_{i, 2345, 45}\}
\end{align}
for each $i \in [f]$. Considering then that the cache content of each user $k \in [5]$ is filled as
\begin{equation}
  \mathcal{Z}_{k} = \{W_{i, \mathcal{S}, \mathcal{T}}: i \in [f], \mathcal{S} \subseteq [5], |\mathcal{S}| = 4, \mathcal{T} \subseteq \mathcal{S}, |\mathcal{T}| = 2, k \in \mathcal{S} \cap \mathcal{T}\}
\end{equation}
it can be easily seen that each user desires a total of $\binom{\alpha - 1}{t} = \binom{3}{2} = 3$ subfiles, each of size $B/6$ bits. Recalling that the vector of ordered users $\hat{\bm{u}} = (\hat{u}_1, \hat{u}_2, \hat{u}_3, \hat{u}_4, \hat{u}_5)$ satisfies $\mathcal{D}_{\hat{u}_{k}} = \cup_{i = k}^{k + 3} \{\hat{u}_{i}\}$ for each $k \in [5]$, we conclude that the subfiles desired by each user are the following.
\begin{itemize}
    \item User $\hat{u}_1$ desires the subfiles $W_{f_{\hat{u}_1}, \mathcal{D}_{\hat{u}_1}, \{\hat{u}_2, \hat{u}_3\}}$, $W_{f_{\hat{u}_1}, \mathcal{D}_{\hat{u}_1}, \{\hat{u}_2, \hat{u}_4\}}$ and $W_{f_{\hat{u}_1}, \mathcal{D}_{\hat{u}_1}, \{\hat{u}_3, \hat{u}_4\}}$.
    \item User $\hat{u}_2$ desires the subfiles $W_{f_{\hat{u}_2}, \mathcal{D}_{\hat{u}_2}, \{\hat{u}_3, \hat{u}_4\}}$, $W_{f_{\hat{u}_2}, \mathcal{D}_{\hat{u}_2}, \{\hat{u}_3, \hat{u}_5\}}$ and $W_{f_{\hat{u}_2}, \mathcal{D}_{\hat{u}_2}, \{\hat{u}_4, \hat{u}_5\}}$.
    \item User $\hat{u}_3$ desires the subfiles $W_{f_{\hat{u}_3}, \mathcal{D}_{\hat{u}_3}, \{\hat{u}_1, \hat{u}_4\}}$, $W_{f_{\hat{u}_3}, \mathcal{D}_{\hat{u}_3}, \{\hat{u}_1, \hat{u}_5\}}$ and $W_{f_{\hat{u}_3}, \mathcal{D}_{\hat{u}_3}, \{\hat{u}_4, \hat{u}_5\}}$.
    \item User $\hat{u}_4$ desires the subfiles $W_{f_{\hat{u}_4}, \mathcal{D}_{\hat{u}_4}, \{\hat{u}_1, \hat{u}_2\}}$, $W_{f_{\hat{u}_4}, \mathcal{D}_{\hat{u}_4}, \{\hat{u}_1, \hat{u}_5\}}$ and $W_{f_{\hat{u}_4}, \mathcal{D}_{\hat{u}_4}, \{\hat{u}_2, \hat{u}_5\}}$.
    \item User $\hat{u}_5$ desires the subfiles $W_{f_{\hat{u}_5}, \mathcal{D}_{\hat{u}_5}, \{\hat{u}_1, \hat{u}_2\}}$, $W_{f_{\hat{u}_5}, \mathcal{D}_{\hat{u}_5}, \{\hat{u}_1, \hat{u}_3\}}$ and $W_{f_{\hat{u}_5}, \mathcal{D}_{\hat{u}_5}, \{\hat{u}_2, \hat{u}_3\}}$. 
\end{itemize}
These subfiles are delivered by the following sequence of XORs
\begin{align}
    X_1 & = \Wtwo{1}{2}{3} \oplus \Wtwo{2}{3}{5} \oplus \Wtwo{3}{1}{4} \\
    X_2 & = \Wtwo{3}{1}{4} \oplus \Wtwo{1}{2}{4} \oplus \Wtwo{4}{1}{2} \\
    X_3 & = \Wtwo{2}{3}{5} \oplus \Wtwo{5}{1}{3} \oplus \Wtwo{3}{1}{5} \\
    X_4 & = \Wtwo{1}{3}{4} \oplus \Wtwo{4}{1}{5} \\
    X_5 & = \Wtwo{2}{3}{4} \oplus \Wtwo{4}{2}{5} \\
    X_6 & = \Wtwo{2}{4}{5} \oplus \Wtwo{5}{1}{2} \\
    X_7 & = \Wtwo{3}{4}{5} \oplus \Wtwo{5}{2}{3}
\end{align}
after which each user $k \in [5]$ can employ its own cache content $\mathcal{Z}_k$ to decode as follows.
\begin{itemize}
    \item User $\hat{u}_1$ recovers its desired subfiles from $X_1 \oplus X_3$, $X_2$ and $X_4$.
    \item User $\hat{u}_2$ recovers its desired subfiles from $X_1 \oplus X_2$, $X_5$ and $X_6$.
    \item User $\hat{u}_3$ recovers its desired subfiles from $X_1$, $X_3$ and $X_7$.
    \item User $\hat{u}_4$ recovers its desired subfiles from $X_2$, $X_4$ and $X_5$.
    \item User $\hat{u}_5$ recovers its desired subfiles from $X_3$, $X_6$ and $X_7$.
\end{itemize}
Given that $|X_i| = B/6$ for each $i \in [7]$, and given that there are $7$ transmissions, we have a load of $R(2) = |X|/B = 7/6$. Since then $R_{\text{LB}}(2) = 7/6$, we can conclude that the converse is tight.

\subsubsection{The case of \texorpdfstring{$t = 3$}{t = 3}}

In this case each file is split into $\binom{4}{3} = 4$ non-overlapping subfiles as
\begin{align}
    W_{i, 1234} & = \{W_{i, 1234, 123}, W_{i, 1234, 124}, W_{i, 1234, 134}, W_{i, 1234, 234}\} \\
    W_{i, 1235} & = \{W_{i, 1235, 123}, W_{i, 1235, 125}, W_{i, 1235, 135}, W_{i, 1235, 235}\} \\
    W_{i, 1245} & = \{W_{i, 1245, 124}, W_{i, 1245, 125}, W_{i, 1245, 145}, W_{i, 1245, 245}\} \\
    W_{i, 1345} & = \{W_{i, 1345, 134}, W_{i, 1345, 135}, W_{i, 1345, 145}, W_{i, 1345, 345}\} \\
    W_{i, 2345} & = \{W_{i, 2345, 234}, W_{i, 2345, 235}, W_{i, 2345, 245}, W_{i, 2345, 345}\}
\end{align}
for each $i \in [f]$. Each user then desires $\binom{3}{3} = 1$ subfile of size $B/4$. More precisely, always considering the general circular demands identified by the vector $\hat{\bm{u}}$, the desired subfiles are given as follows.
\begin{itemize}
    \item User $\hat{u}_1$ desires $\Wthree{1}{2}{3}{4}$.
    \item User $\hat{u}_2$ desires $\Wthree{2}{3}{4}{5}$.
    \item User $\hat{u}_3$ desires $\Wthree{3}{1}{4}{5}$.
    \item User $\hat{u}_4$ desires $\Wthree{4}{1}{2}{5}$.
    \item User $\hat{u}_5$ desires $\Wthree{5}{1}{2}{3}$.
\end{itemize}
After transmitting the following two XORs
\begin{align}
    X_1 & = \Wthree{1}{2}{3}{4} \oplus \Wthree{2}{3}{4}{5} \oplus \Wthree{4}{1}{2}{5} \\
    X_2 & = \Wthree{1}{2}{3}{4} \oplus \Wthree{3}{1}{4}{5} \oplus \Wthree{5}{1}{2}{3} 
\end{align}
each user can decode as follows.
\begin{itemize}
    \item User $\hat{u}_1$ and user $\hat{u}_3$ recover their desired subfiles from $X_2$.
    \item User $\hat{u}_2$ and user $\hat{u}_4$ recover their desired subfiles from $X_1$.
    \item User $\hat{u}_5$ recovers its desired subfile from $X_1 \oplus X_2$.
\end{itemize}
Recalling that $|X_1| = |X_2| = B/4$, the $2$ transmissions correspond to a communication load $R(3) = |X|/B = 1/2$ which matches the converse $R_{\text{LB}}(3) = 1/2$. This means that the scheme is optimal among all the caching-and-delivery schemes that deliver circular demands.

\subsection{Circular Demands and the \texorpdfstring{$(6, 5, f)$}{(5, 4, f)} FDS Structure}

In this setting we consider the $(K, \alpha, f) = (6, 5, f)$ structure. Here, there are $C = 6$ classes of files $\mathcal{W}_{12345}, \mathcal{W}_{12346}, \mathcal{W}_{12356}, \mathcal{W}_{12456}, \mathcal{W}_{13456}, \mathcal{W}_{23456}$ and $N = fC = 6f$ files in total. The $6$ FDSs take the form
\begin{align}
    \mathcal{F}_{1} &= \{\mathcal{W}_{12345}, \mathcal{W}_{12346}, \mathcal{W}_{12356}, \mathcal{W}_{12456}, \mathcal{W}_{13456}\} \\
    \mathcal{F}_{2} &= \{\mathcal{W}_{12345}, \mathcal{W}_{12346}, \mathcal{W}_{12356}, \mathcal{W}_{12456}, \mathcal{W}_{23456}\} \\
    \mathcal{F}_{3} &= \{\mathcal{W}_{12345}, \mathcal{W}_{12346}, \mathcal{W}_{12356}, \mathcal{W}_{13456}, \mathcal{W}_{23456}\} \\
    \mathcal{F}_{4} &= \{\mathcal{W}_{12345}, \mathcal{W}_{12346}, \mathcal{W}_{12456}, \mathcal{W}_{13456}, \mathcal{W}_{23456}\} \\
    \mathcal{F}_{5} &= \{\mathcal{W}_{12345}, \mathcal{W}_{12356}, \mathcal{W}_{12456}, \mathcal{W}_{13456}, \mathcal{W}_{23456}\} \\
    \mathcal{F}_{6} &= \{\mathcal{W}_{12346}, \mathcal{W}_{12356}, \mathcal{W}_{12456}, \mathcal{W}_{13456}, \mathcal{W}_{23456}\} 
\end{align}
where $\mathcal{W}_{\mathcal{S}} = \{W_{i, \mathcal{S}} : i \in [f]\}$. As in the previous case, we here provide a scheme for any $f \in \mathbb{Z}^{+}$ and any circular demand. Each such circular demand is identified by a vector $\hat{\bm{u}} = (\hat{u}_1, \hat{u}_2, \hat{u}_3, \hat{u}_4, \hat{u}_5, \hat{u}_6)$, and it induces the FDS request graph in~\Cref{fig: FDS Request Graph for a Generic Circular Demand and Second FDS Structure}. The optimal scheme is provided for the case $t = 3$.

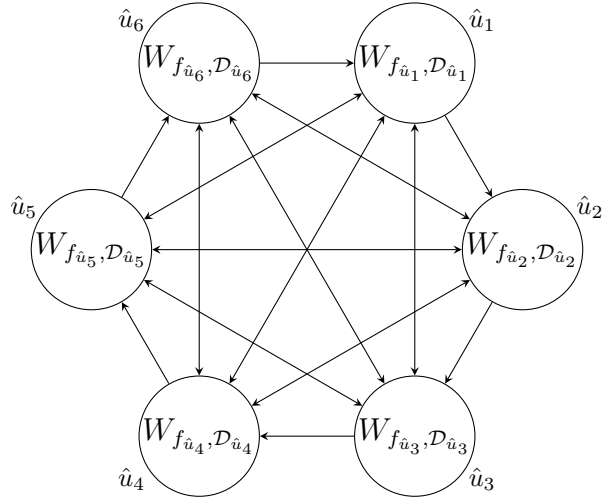
\begin{figure}[!htb]
\centering
\begin{tikzpicture}
  \node[regular polygon, regular polygon sides = 6, inner sep = 1.75cm](G){};
  \node[draw, circle, label = center:$W_{f_{\hat{u}_1}, \mathcal{D}_{\hat{u}_1}}$, minimum size = 1.6cm](1) at (G.corner 1){};
  \node[draw, circle, label = center:$W_{f_{\hat{u}_6}, \mathcal{D}_{\hat{u}_6}}$, minimum size = 1.6cm](6) at (G.corner 2){};
  \node[draw, circle, label = center:$W_{f_{\hat{u}_5}, \mathcal{D}_{\hat{u}_5}}$, minimum size = 1.6cm](5) at (G.corner 3){};
  \node[draw, circle, label = center:$W_{f_{\hat{u}_4}, \mathcal{D}_{\hat{u}_4}}$, minimum size = 1.6cm](4) at (G.corner 4){};
  \node[draw, circle, label = center:$W_{f_{\hat{u}_3}, \mathcal{D}_{\hat{u}_3}}$, minimum size = 1.6cm](3) at (G.corner 5){};
  \node[draw, circle, label = center:$W_{f_{\hat{u}_2}, \mathcal{D}_{\hat{u}_2}}$, minimum size = 1.6cm](2) at (G.corner 6){};
  \node at (1.north east)[font = \small, anchor = west]{$\hat{u}_1$};
  \node at (2.north east)[font = \small, anchor = west]{$\hat{u}_2$};
  \node at (3.south east)[font = \small, anchor = west]{$\hat{u}_3$};
  \node at (4.south west)[font = \small, anchor = east]{$\hat{u}_4$};
  \node at (5.north west)[font = \small, anchor = east]{$\hat{u}_5$};
  \node at (6.north west)[font = \small, anchor = east]{$\hat{u}_6$};
  \draw[-stealth](1)--(2); \draw[-stealth](2)--(3); \draw[-stealth](3)--(4); \draw[-stealth](4)--(5); \draw[-stealth](5)--(6); \draw[-stealth](6)--(1);
  \draw[stealth-stealth](1)--(3); \draw[stealth-stealth](1)--(4); \draw[stealth-stealth](1)--(5);
  \draw[stealth-stealth](2)--(4); \draw[stealth-stealth](2)--(5); \draw[stealth-stealth](2)--(6);
  \draw[stealth-stealth](3)--(5); \draw[stealth-stealth](3)--(6); 
  \draw[stealth-stealth](4)--(6);
\end{tikzpicture}
\caption{FDS request graph for a generic circular demand identified by the vector $\hat{\bm{u}} = (\hat{u}_1, \hat{u}_2, \hat{u}_3, \hat{u}_4, \hat{u}_5, \hat{u}_6)$ and the $(K, \alpha, f) = (6, 5, f)$ FDS structure.}
\label{fig: FDS Request Graph for a Generic Circular Demand and Second FDS Structure}
\end{figure}

According to the cache placement in \Cref{sec: Cache Placement Procedure}, each file is split into $\binom{\alpha}{t} = \binom{5}{3} = 10$ non-overlapping subfiles as
\begin{align}
    W_{i, \mathcal{S}} & = \{W_{i, \mathcal{S}, \mathcal{T}} : \mathcal{T} \subseteq \mathcal{S}, |\mathcal{T}| = 3\}, \quad \forall \mathcal{S} \subseteq [6] : |\mathcal{S}| = 5, \quad \forall i \in [f]
\end{align}
where each subfile has size $B/10$. For example, the file $W_{i, 12345}$ is split into $10$ non-overlapping subfiles labeled as $W_{i, 12345, \mathcal{T}}$ for each $\mathcal{T} \in \{123, 124, 125, 134, 135, 145, 234, 235, 245, 345\}$. We recall that the set $\mathcal{T}$ represents the users which the subfile $W_{i, 12345, \mathcal{T}}$ is exactly and uniquely cached at. If we consider a generic circular demand, each user misses $\binom{\alpha - 1}{t} = \binom{4}{3} = 4$ subfiles given by the following.
\begin{itemize}
    \item User $\hat{u}_1$ desires the subfiles $\Wthree{1}{2}{3}{4}$, $\Wthree{1}{2}{3}{5}$, $\Wthree{1}{2}{4}{5}$ and $\Wthree{1}{3}{4}{5}$.
    \item User $\hat{u}_2$ desires the subfiles $\Wthree{2}{3}{4}{5}$, $\Wthree{2}{3}{4}{6}$, $\Wthree{2}{3}{5}{6}$ and $\Wthree{2}{4}{5}{6}$.
    \item User $\hat{u}_3$ desires the subfiles $\Wthree{3}{1}{4}{5}$, $\Wthree{3}{1}{4}{6}$, $\Wthree{3}{1}{5}{6}$ and $\Wthree{3}{4}{5}{6}$.
    \item User $\hat{u}_4$ desires the subfiles $\Wthree{4}{1}{2}{5}$, $\Wthree{4}{1}{2}{6}$, $\Wthree{4}{1}{5}{6}$ and $\Wthree{4}{2}{5}{6}$.
    \item User $\hat{u}_5$ desires the subfiles $\Wthree{5}{1}{2}{3}$, $\Wthree{5}{1}{2}{6}$, $\Wthree{5}{1}{3}{6}$ and $\Wthree{5}{2}{3}{6}$.
    \item User $\hat{u}_6$ desires the subfiles $\Wthree{6}{1}{2}{3}$, $\Wthree{6}{1}{2}{4}$, $\Wthree{6}{1}{3}{4}$ and $\Wthree{6}{2}{3}{4}$.
\end{itemize}
If we consider the following linear combinations of subfiles
\begin{align}
    X_1 & = \Wthree{1}{2}{3}{4} \oplus \Wthree{2}{4}{5}{6} \oplus \Wthree{4}{2}{5}{6} \\
    X_2 & = \Wthree{4}{1}{5}{6} \oplus \Wthree{1}{2}{3}{5} \oplus \Wthree{5}{1}{2}{3} \\
    X_3 & = \Wthree{1}{2}{3}{4} \oplus \Wthree{4}{1}{5}{6} \oplus \Wthree{6}{1}{3}{4} \oplus \Wthree{3}{1}{4}{6} \\
    X_4 & = \Wthree{2}{3}{4}{5} \oplus \Wthree{3}{1}{5}{6} \oplus \Wthree{5}{1}{3}{6} \\
    X_5 & = \Wthree{5}{1}{2}{6} \oplus \Wthree{2}{3}{4}{6} \oplus \Wthree{6}{2}{3}{4} \\
    X_6 & = \Wthree{2}{3}{4}{5} \oplus \Wthree{5}{1}{2}{6} \oplus \Wthree{1}{2}{4}{5} \oplus \Wthree{4}{1}{2}{5} \\
    X_7 & = \Wthree{3}{4}{5}{6} \oplus \Wthree{4}{1}{2}{6} \oplus \Wthree{6}{1}{2}{4} \\
    X_8 & = \Wthree{6}{1}{2}{3} \oplus \Wthree{3}{1}{4}{5} \oplus \Wthree{1}{3}{4}{5} \\
    X_9 & = \Wthree{3}{4}{5}{6} \oplus \Wthree{6}{1}{2}{3} \oplus \Wthree{2}{3}{5}{6} \oplus \Wthree{5}{2}{3}{6}
\end{align}
and we denote by $X = (X_i : i \in [9])$ the concatenated message sent by the central server, then each user can correctly decode its desired subfiles as follows.
\begin{itemize}
    \item User $\hat{u}_1$ recovers its desired subfiles from $X_2$, $X_3$, $X_4 \oplus X_6$ and $X_8$.
    \item User $\hat{u}_2$ recovers its desired subfiles from $X_1$, $X_5$, $X_6$ and $X_7 \oplus X_9$.
    \item User $\hat{u}_3$ recovers its desired subfiles from $X_2 \oplus X_3$, $X_4$, $X_8$ and $X_9$.
    \item User $\hat{u}_4$ recovers its desired subfiles from $X_1$, $X_3$, $X_5 \oplus X_6$ and $X_7$.
    \item User $\hat{u}_5$ recovers its desired subfiles from $X_2$, $X_4$, $X_6$ and $X_8 \oplus X_9$.
    \item User $\hat{u}_6$ recovers its desired subfiles from $X_1 \oplus X_3$, $X_5$, $X_7$ and $X_9$.
\end{itemize}
The delivery procedure is slightly more involved with respect to the previous FDS structure. Indeed, the messages $X_i$ are carefully designed in such a way that also their linear combinations can be useful to some users. An equivalent interpretation of this fact is related to the previously mentioned creation of cliques. Consider for example the XOR $X_3$. User~$1$ and user~$4$ can directly cache-out interference to correctly decode their desired subfiles, while user~$3$ and user~$6$ miss in their cache --- due to the selfish cache placement --- some interfering messages appearing in the XOR $X_3$. Such interfering (and consequently undesired) messages are ``delivered'' to both user~$3$ and user~$6$ by means of XORs $X_2$ and $X_1$, hence allowing them to decode the desired subfiles from $X_3$. We can see here that with $X_3$ we are able serve the clique composed by user~$1$, user~$4$, user~$3$ and user~$6$, by carefully ``passing'' some undesired (and not cached) information to the last two users. A similar reasoning applies to the XORs $X_6$ and $X_9$, both of which are useful to $4$ users simultaneously.

The communication load is equal to $R(3) = |X|/B = 9/10$ and it matches the converse $R_{\text{LB}}(3) = 9/10$. Hence, the converse here is tight.

\section{Conclusions}\label{sec: Conclusion}

In this work, we investigated the effects that selfish caching can have on the optimal worst-case communication load in the coded caching framework. The proposed general FDS structure seeks to capture the degree of intersection between the interests of the different users. While somewhat restrictive, the proposed structure was designed to bring to the fore and accentuate the adversarial relationship between coded caching and selfish caching, and by doing so, to allow us to provide insight on the nature of this adversarial relationship.

This insight is provided with the introduction here of a new information-theoretic converse on the minimum worst-case communication load by means of index coding arguments. For the proposed broad FDS structure, the converse bound definitively resolves the question of whether selfish caching is generally beneficial or not. More specifically, the converse reveals that any non-zero load brought about by symmetrically selfish caching is always (with the exception of the extreme points of $t$) strictly worse than the optimal load guaranteed in the unselfish scenario.  The rationale behind this is that, despite the sizeable increase of local caching gain brought about by the very targeted placement of selfish caching, and despite a very restricted set of demands, the loss in multicasting opportunities is too severe. In fact, what the converse shows is that this damage is so prominent that --- for any fixed (or decreasing) ratio $\delta= \alpha/K < 1$ --- the coding gain does not scale with $K$, and is in fact bounded above by $1/(1 - \delta)$. In other words, even if there is, for example, a $\SI{99}{\percent}$ symmetric intersection between the interests of the users, the coding gain will not scale as $K$ increases.

The above realisation brings to the fore several interesting research directions for the future. One possibility is to study new, less restrictive FDS structures that could allow more flexibility in designing multicasting opportunities. One way to do this is to consider that each user does not have its own FDS. Clearly, a remaining challenge is to provide optimal schemes for any set of demands, either for the proposed or for another FDS structure.

\appendices

\section{Proof of \texorpdfstring{\Cref{lem: Circular Shift Lemma}}{Lemma~\ref{lem: Circular Shift Lemma}}}\label{app: Proof of the Circular Shift Lemma}

Let $\hat{\bm{u}} = (\hat{u}_{1}, \dots, \hat{u}_{K})$ be a permutation of the elements in $[K]$. Consider $k_1, k_2 \in [K]$ such that $k_1 \neq k_2$. Assume without loss of generality that $\hat{u}(k_1) < \hat{u}(k_2) \leq K$, in which case $\ell = \hat{u}(k_2) - \hat{u}(k_1)$. Denoting by $\mathcal{U}$ the set containing the $K$ circular shifts of the vector $\hat{\bm{u}}$, we see that there are $\ell$ vectors $\bm{u} \in \mathcal{U}$ such that $k_2$ appears before $k_1$ in $\bm{u}$. These cases correspond to the vectors $\bm{u} \in \mathcal{U}$ where we have in the first position of $\bm{u}$ either the element $k_2$, or any one of the $(\ell - 1)$ elements between $\hat{u}(k_1)$ and $\hat{u}(k_2)$ in the vector $\hat{\bm{u}}$. As a consequence, the total number of vectors $\bm{u} \in \mathcal{U}$ such that $k_1$ appears before $k_2$ is equal to $(K - \ell)$, which concludes the proof.

\section{Proof of \texorpdfstring{\Cref{lem: Strictly Decreasing Sequence}}{Lemma~\ref{lem: Strictly Decreasing Sequence}}}\label{app: Proof of the Strictly Decreasing Sequence Lemma}

The convexity of $f(t)$ can be easily shown by verifying that the second derivative $f''(t)$ with respect to $t$ is strictly positive for $t \geq 0$. Indeed, we have
\begin{align}
    f(t) & = \frac{\binom{\alpha}{t + 1} + (K - \alpha)\binom{\alpha - 1}{t}}{\binom{\alpha}{t}} = \frac{\alpha - t}{1 + t} + (K - \alpha)\left(1 - \frac{t}{\alpha}\right) \\
    f'(t) & = -\frac{1 + \alpha}{(1 + t)^2} - \frac{(K - \alpha)}{\alpha} \\
    f''(t) & = \frac{2(1 + \alpha)}{(1 + t)^3} > 0
\end{align}
where $f'(t)$ denotes the first derivative. 
Then, since $t \in [0 : \alpha]$, we can evaluate $f(0) = K$ and $f(\alpha) = 0$, showing that $f(0) > f(\alpha)$. Hence, since $f(t)$ is convex, it has to be also strictly decreasing for $t \in [0 : \alpha]$, otherwise the convexity property would be violated. This concludes the proof.

\section{Proof of \texorpdfstring{\Cref{cor: Comparison between Converse Bound and MAN Scheme}}{\ref{cor: Comparison between Converse Bound and MAN Scheme}}}\label{app: Proof of Comparison with MAN}

Recalling that $R_{\text{MAN}}(t) = \frac{\binom{K}{t + 1}}{\binom{K}{t}}$ as well as recalling the load expression for $R_{\text{LB}}(t)$ in \Cref{thm: Lower Bound for the FDS Structure}, we have that
\begin{align}
\frac{R^{\star}(t)}{R_{\text{MAN}}(t)} \geq \frac{R_{\text{LB}}(t)}{R_{\text{MAN}}(t)} & = \frac{\binom{\alpha}{t + 1} + (K - \alpha)\binom{\alpha - 1}{t}}{\binom{\alpha}{t}} \cdot \frac{\binom{K}{t}}{\binom{K}{t + 1}} \\
                                           & = \frac{\binom{\alpha}{t + 1}}{\binom{\alpha}{t}} \cdot \left(1 + (K - \alpha) \frac{t + 1}{\alpha}\right) \cdot \frac{\binom{K}{t}}{\binom{K}{t + 1}} \\
                                           & = \frac{\alpha - t}{K - t} \cdot \left(\frac{K(1 + t) - \alpha t}{\alpha}\right) \\
                                           & = 1 + \underbrace{\frac{t(K - \alpha)(\alpha - 1 - t)}{\alpha(K - t)}}_{\geq 0}
\end{align}
where the second term in the last expression is equal to $0$ either when $\alpha \in [K - 1]$ for $t \in \{0, \alpha - 1\}$, or when $\alpha = K$ and $f \geq K$ for any $t$. This concludes the proof.

\section{Proof of \texorpdfstring{\Cref{cor: Bound on the Coding Gain}}{Corollary~\ref{cor: Bound on the Coding Gain}}}\label{app: Proof of the Bound on the Coding Gain}

We know that the optimal coding gain is upper bounded as
\begin{equation}
    G^\star \leq \frac{t + 1}{1 + \frac{t}{K}(K - \alpha)} = \bar{G}(t).
\end{equation}
It can be easily verified that $\bar{G}''(t) < 0$ for $t \geq 0$, which means that $\bar{G}(t)$ is concave for positive values of $t$. Then, we can see that $\bar{G}(0) = 1$, whereas
\begin{equation}
    \lim_{t \to \infty} \bar{G}(t) = \frac{K}{K - \alpha} > 1
\end{equation}
for any $\alpha \in [K - 1]$. Consequently, $\bar{G}(t) < K/(K - \alpha)$, which means that $G^\star < 1/(1 - \delta)$ for any $\delta = \alpha/K$. This concludes the proof.

\section{Converse Proof of \texorpdfstring{\Cref{thm: The Exact Worst-Case Load for alpha-Demands}}{Proposition \ref{thm: The Exact Worst-Case Load for alpha-Demands}}}\label{app: Converse Proof of the Exact Worst-Case Load for alpha-Demands}

While the achievable expression matches exactly the converse expression $R_{\text{LB}}$, this latter converse cannot be used to prove the optimality of $R^{\star}_{\alpha, \text{c}}$, because $R_{\text{LB}}$ bounds the optimal \emph{worst-case} communication load. As there is no a priori guarantee that the $\alpha$-demands are part of the worst-case demands, we will here derive another bound that focuses on $\alpha$-demands to prove that the achievable performance is indeed optimal.

Following the same line of reasoning as in \Cref{sec: Proof of the Lower Bound}, we apply again the index coding lower bound in \Cref{thm: Acyclic Subgraph Converse Bound}, with the only difference being that now the cache placement is fixed. The corresponding index coding problem has $K' = K$ users and $N' = K\binom{\alpha - 1}{t}$ messages, where $\binom{\alpha - 1}{t}$ is the total number of subfiles desired by each user for a fixed value of $t \in [0 : \alpha]$. The desired message set and the side information set are respectively given by
\begin{align}
  \mathcal{M}_{k} & = \{W_{f_{k}, \mathcal{D}_{k}, \mathcal{T}}: \mathcal{T} \subseteq \mathcal{D}_{k}, |\mathcal{T}| = t, k \notin \mathcal{T}\} \\
  \mathcal{A}_{k} & = \{W_{i, \mathcal{S}, \mathcal{T}} : i \in [f], \mathcal{S} \subseteq [K], |\mathcal{S}| = \alpha, \mathcal{T} \subseteq \mathcal{S}, |\mathcal{T}| = t, k \in \mathcal{S} \cap \mathcal{T}\}
\end{align}
for all $k \in [K]$. In the corresponding side information graph, an edge exists from $W_{f_{k_{1}}, \mathcal{D}_{k_{1}}, \mathcal{T}_{1}}$ to $W_{f_{k_{2}}, \mathcal{D}_{k_{2}}, \mathcal{T}_{2}}$ if and only if $W_{f_{k_{1}}, \mathcal{D}_{k_{1}}, \mathcal{T}_{1}} \in \mathcal{A}_{k_{2}}$.

Since we are considering a converse bound on the optimal communication load under a specific cache placement and under a specific set of demands, it suffices to find a single $\alpha$-demand such that $R^{\star}_{\alpha, \text{c}}(t) \geq \frac{\binom{\alpha}{t + 1} + (K - \alpha)\binom{\alpha - 1}{t}}{\binom{\alpha}{t}}$. Toward this, consider the $\alpha$-demand where $\mathcal{K}$ is a set of $\alpha$ users that request distinct files from the file class $\mathcal{K}$, and where the remaining users in $[K] \setminus \mathcal{K}$ request distinct files so that the set of vertices
\begin{equation}
  \mathcal{J}_{1} = \bigcup_{k \in [K] \setminus \mathcal{K}} \bigcup_{\mathcal{T} \subseteq \mathcal{D}_{k} \setminus \{k\} : |\mathcal{T}| = t} \{W_{f_{k}, \mathcal{D}_{k}, \mathcal{T}}\}
\end{equation}
is acyclic. Then, such set contains a total of $(K - \alpha)\binom{\alpha - 1}{t}$ subfiles, which means that $|\mathcal{J}_1| = (K - \alpha)\binom{\alpha - 1}{t}\frac{B}{\binom{\alpha}{t}}$. Indeed, we can see that there exist $\alpha$-demands for which $\mathcal{J}_{1}$ is acyclic. For instance, if we assume $\mathcal{D}_{k} =  \mathcal{S} \cup \{k\}$ for every $k \in [K] \setminus \mathcal{K}$ for some $\mathcal{S} \subseteq \mathcal{K}$ such that $|\mathcal{S}| = \alpha - 1$, then the set $\mathcal{J}_{1}$ is acyclic.

Consider now the set of users in $\mathcal{K}$. Take any permutation of users $\bm{u} = (u_{1}, \dots, u_{\alpha})$ with $u_{k} \in \mathcal{K}$ for all $k \in [\alpha]$. Then, since $\mathcal{D}_{k} = \mathcal{K}$ for all $k \in \mathcal{K}$, the set of vertices
\begin{equation}
  \mathcal{J}_{2} = \bigcup_{k \in [\alpha]} \bigcup_{\mathcal{T} \subseteq \mathcal{K} \setminus \{u_{1}, \dots, u_{k}\} : |\mathcal{T}| = t}\{W_{f_{u_{k}}, \mathcal{K}, \mathcal{T}}\}
\end{equation}
is acyclic for any permutation $\bm{u}$ (see \cite[Lemma 1]{8963629}). It can be easily seen that such set contains a total of $\binom{\alpha - 1}{t} + \binom{\alpha - 2}{t} + \dots + \binom{t}{t} = \binom{\alpha}{t + 1}$ subfiles, hence $|\mathcal{J}_2| = \binom{\alpha}{t + 1}\frac{B}{\binom{\alpha}{t}}$.

Due to the fact that $\mathcal{D}_{k} = \mathcal{K}$ for all $k \in \mathcal{K}$, there is no edge connecting any vertex in $\mathcal{J}_{2}$ to any vertex in $\mathcal{J}_{1}$, therefore also the set $\mathcal{J}_{1} \cup \mathcal{J}_{2}$ is acyclic. At this point, applying \Cref{thm: Acyclic Subgraph Converse Bound} with respect to the acyclic set $\mathcal{J}_{1} \cup \mathcal{J}_{2}$, we get
\begin{align}
  BR^{\star}_{\alpha, \text{c}}  &\geq \sum_{k \in [K] \setminus \mathcal{K}} \sum_{\mathcal{T} \subseteq \mathcal{D}_{k} \setminus \{k\} : |\mathcal{T}| = t} \left| W_{f_{k}, \mathcal{D}_{k}, \mathcal{T}} \right| + \sum_{k \in [\alpha]} \sum_{\mathcal{T} \subseteq \mathcal{K} \setminus \{u_{1}, \dots, u_{k}\} : |\mathcal{T}| = t}\left| W_{f_{u_{k}}, \mathcal{K}, \mathcal{T}} \right| \\
                                    & = |\mathcal{J}_1| + |\mathcal{J}_2| \\
                                    & = (K - \alpha)\binom{\alpha - 1}{t}\frac{B}{\binom{\alpha}{t}} + \binom{\alpha}{t + 1}\frac{B}{\binom{\alpha}{t}}
\end{align}
which means that $R_{\alpha, \text{c}}^{\star}(t) \geq \frac{\binom{\alpha}{t + 1} + (K - \alpha)\binom{\alpha - 1}{t}}{\binom{\alpha}{t}}$. This concludes the proof.

\printbibliography

\end{document}